\documentclass [11pt]{article}

\usepackage{amsmath, amsthm}
\usepackage{amssymb}
\usepackage{graphicx}
\usepackage{fullpage}

\newtheorem{fact}{Fact}[section]

\newtheorem{defi}{Definition}[section]
\newtheorem{exam}{Example}[section]
\newtheorem{lemma}{Lemma}[section]
\newtheorem{prop}{Proposition}[section]
\newtheorem{theorem}{Theorem}[section]

\newcommand{\algb}{{\sc Alg-Existence}}
\newcommand{\algmin}{{\sc Alg-Min-Equilibrium}}
\newcommand{\increaseprice}{{\sc Price-Increment}}
\newcommand{\algmax}{{\sc Alg-Max-Equilibrium}}
\newcommand{\ie}{{i.e.}}

\title{Competitive Equilibria in Matching Markets with Budgets}\vspace{0.3in}

\author{
Ning Chen\footnote{Division of Mathematical Sciences, Nanyang Technological
University, Singapore. Email: {\sf ningc@ntu.edu.sg.}} \and Xiaotie
Deng\footnote{Department of Computer Science, City University of
Hong Kong, Hong Kong. Email: {\sf csdeng@cityu.edu.hk.}}
\and Arpita Ghosh\footnote{Yahoo! Research, Santa Clara, CA, USA. Email: {\sf arpita@yahoo-inc.com.}}}

\begin{document}

\maketitle

\begin{abstract}

We study competitive equilibria in the classic Shapley-Shubik assignment model with indivisible goods and unit-demand buyers, with {\em budget constraints}: buyers can specify a maximum price they are willing to pay for each item, beyond which they cannot afford the item.  This single discontinuity introduced by the budget constraint fundamentally changes the properties of equilibria: in the assignment model without budget constraints, a competitive equilibrium always exists, and corresponds exactly to a stable matching. With budgets, a competitive equilibrium need not always exist. In addition, there are now {\em two} distinct notions of stability, depending on whether both or only one of the buyer and seller can strictly benefit in a blocking pair, that no longer coincide due to the budget-induced discontinuity. We define {\em weak} and {\em strong stability} for the assignment model with transferable utilities, and show that competitive equilibria correspond exactly to strongly stable matchings.

We consider the algorithmic question of efficiently computing competitive equilibria in an extension of the assignment model with budgets, where each buyer specifies his preferences over items using {\em utility functions} $u_{ij}$, where $u_{ij}(p_j)$  is the utility of buyer $i$ for item $j$ when its price is $p_j$. Our main result is a strongly polynomial time algorithm that decides whether or not a competitive equilibrium exists and if yes, computes a minimum one, for a general class of utility functions $u_{ij}$. This class of utility functions includes the standard quasi-linear utility model with a budget constraint, and in addition, allows modeling marketplaces where, for example, buyers only have a preference ranking amongst items subject to a maximum payment limit for each item, or where buyers want to optimize return on investment (ROI) instead of a quasi-linear utility and only know items' relative values.
\end{abstract}

\setcounter{page}{0}\thispagestyle{empty}
\newpage

\section{Introduction}

Consider a market with $n$ unit demand buyers and $m$ sellers, each selling one unit of an indivisible good. The buyers specify their preferences over items via utility functions $u_{ij}(p_j)$, which is the utility of buyer $i$ for item $j$ when its price is $p_j$. So far, this is the classic Shapley-Shubik assignment model~\cite{SS} which captures a variety of matching markets including housing markets and ad auctions~\cite{EOS,varian}, except for the extension to general utility functions instead of the quasi-linear utilities in the original model. Shapley and Shubik show that a {\em competitive equilibrium} always exists in their model, and later work~\cite{CK,Quinnzi,Gale} shows that a competitive equilibrium must also exist for the model with general utility functions $u_{ij}(\cdot)$, provided these $u_{ij}(\cdot)$ are strictly decreasing and continuous everywhere.

Now suppose we extend the assignment model with an extra {\em budget constraint}:
a buyer $i$ can specify a maximum price $b_{ij}$ that he is able to pay for item $j$, above which he cannot afford the item. That is, the utility function $u_{ij}(\cdot)$ can now (possibly) have a {\em discontinuity} at $p_j = b_{ij}$.
Budgets are a very real constraint in many marketplaces such as advertising markets, and have led to a spate of recent work on auction design~\cite{cg,bk,ec05,pv,focs08,stoc10,shuchi}, where the addition of the budget constraint, while seemingly innocuous, introduces fundamental new challenges to the problem.
As we will see, the same happens in the Shapley-Shubik assignment model: the discontinuity introduced by the budget constraint is not merely technical, but fundamentally changes the properties of competitive equilibria. First, a competitive equilibrium no longer always exists (Example~\ref{example-no-eq} in Appendix~\ref{appendix-example}). Second, and related to the first, while competitive equilibria in the original model correspond precisely to {\em stable matchings}~\cite{SS}, this is not quite true with budgets since the different notions of stability no longer coincide as in the original model, \ie, there is no longer a single unique notion of stability.

A {\em weakly stable} matching is one where there is no unmatched buyer-seller pair where both the buyer and seller can strictly benefit by trading with each other, where a seller's payoff is the payment he receives for his item. A {\em strongly stable} matching is one where there is no unmatched buyer-seller pair where one party strictly benefits and the other weakly benefits from the deviation (an example of the seller only weakly benefiting is when the buyer $i$ strictly prefers to buy seller $j$'s item at its current price, but not at any higher price). Without the budget-induced discontinuity, these two notions can be shown to be identical given the continuity of $u_{ij}(\cdot)$. However, with the discontinuity at $b_{ij}$ introduced by the budget constraint, they are no longer equivalent: a weakly stable matching does not possess the envy-freeness property of a competitive equilibrium and, as we show in Theorem~\ref{theorem-stability}, it is strongly stable matchings that correspond to competitive equilibria.

A natural question, then, is the following: Given a marketplace where buyers have such general utility functions $u_{ij}(\cdot)$ with budget constraints, is it possible to determine whether a competitive equilibrium exists, and if yes, compute one, for instance a buyer-optimal one, efficiently? This question is an interesting theoretical problem in its own right, given the extensive applications of the Shapley-Shubik model, and is analogous to the problem of computing strongly stable matchings in the Gale-Shapley marriage model when preference lists have ties~\cite{irving} in the more complex setting of equilibrium prices.

Our original motivation, though, comes from matching markets, such as the (online and TV) advertising market, where bidders do not always fit the standard model of quasi-linear utility optimizers. One obvious example of this is bidders in advertising markets with budget constraints; these cannot be captured by quasi-linear utility models. As another example, bidders in advertising markets may not be able to accurately estimate their values for items, but only be able to specify a preference ordering amongst items along with a maximum payment limit for an item -- this is because the precise value of placing an ad depends on factors such as conversion rates which are difficult to estimate, while determining the relative ordering amongst different placement options is much easier. Such bidders cannot specify their values for each item in the quasi-linear utility model as well. A third kind of bidders that do not fit the standard model are advertisers who want to optimize return on investment (ROI), which is value divided by price rather than value minus price; again, these advertisers might also be able to better estimate the ratios of values rather than the values themselves. An algorithm that computes a fair and efficient assignment of items with such general inputs can significantly improve the advertising marketplace by making it much easier for advertisers to participate and bid in the market. Of course, clearing a real advertising marketplace would require solving a model with multiple units of supply and demand for each seller and buyer --- we take the first steps towards this rather difficult problem by solving the market clearing problem for the unit demand case.

The remainder of the paper is organized as follows: We first discuss our main result in Section~\ref{intro-main}, and related work in Section~\ref{intro-related},~\ref{intro-market}. In Section~\ref{section-pre}, we define our model and solution concept, and in Section~\ref{section-stability} we prove that competitive equilibria correspond to strongly stable matchings in the assignment model with budgets. In Section~\ref{section-min}, we present our strongly polynomial time algorithm to compute a competitive equilibrium and outline its proof of correctness. Finally, in Section~\ref{section-min-mechanism}, we address the issue of strategic behavior in the marketplace.

\subsection{Main Result}\label{intro-main}

Our main result answers the question of computing competitive equilibria, if they exist, for a general class of utility functions $u_{ij}(\cdot)$ which are continuous and strictly decreasing on $[0, b_{ij}]$ as in \cite{CK,DG}, and satisfy an additional mutual {\em consistency} condition that allows increasing prices in a way that guarantees strongly polynomial runtime. This class of utility functions, which we will call {\em consistent} utility functions, is quite general and models, in addition to the standard quasi-linear utility model with a budget constraint, marketplaces where buyers who only have a preference ranking amongst items subject to a maximum payment limit for each item, or where buyers want to optimize ROI and only know items' relative values.

\bigskip \noindent \textbf{Theorem.}
Suppose we are given an instance of the assignment model with consistent utility functions $u_{ij}$. Then, if a competitive equilibrium exists, a minimum competitive equilibrium exists as well; further, the problem of deciding whether or not an equilibrium exists, and computing a minimum one, can be solved in strongly polynomial time.

\bigskip
Note that for arbitrary utility functions, even when competitive equilibria exist, a minimum one (Definition~\ref{def-min-eq}) need not.
For example, if utility functions are only weakly rather than strictly decreasing, there exist instances (Example~\ref{example-no-min-eq}) where competitive equilibria exist but the associated price vectors are {\em not comparable}, so no minimum equilibrium exists.

The constructive proof of this result is provided by algorithm \algmin, which returns a minimum competitive equilibrium if any competitive equilibrium exists for the given instance, or reports no equilibrium exists. Starting with the zero price vector $\mathbf{p}=0$, the algorithm constructs a bipartite {\em dynamic demand graph} $G^+(\mathbf{p})$ based on the demand sets --- the set of items with maximal, positive, utility at the current prices $\mathbf{p}$. It then identifies the set of ``over-demanded" items in this demand graph, which is captured by the critical set $A$ of buyers and its neighborhood $N(A)$ in $G^+(\mathbf{p})$: not all buyers in $A$ can be matched to distinct items in $N(A)$. Therefore, there can be no equilibrium at price $\mathbf{p}$, and the prices of items in $N(A)$ need to be raised until these items are no longer over-demanded. The algorithm recursively increases the prices of over-demanded items using a subroutine \increaseprice. Eventually, after all critical sets have been eliminated, the algorithm checks whether there is a matching that satisfies the conditions for a competitive equilibrium, namely all buyers are assigned an item in their demand set and every item with price greater than zero is assigned to a buyer.

There are two key challenges to developing a strongly polynomial time algorithm that returns a minimum competitive equilibrium for general utility functions $u_{ij}$ with budget constraints. The first comes from the fact that we allow rather general utility functions. The algorithm needs to increase prices of over-demanded items in each step to the maximum extent possible to ensure a fast runtime, while making sure that prices do not overshoot any equilibrium price vectors. However, to ensure that the algorithm does successfully find a minimum equilibrium or else correctly reports that no equilibrium exists, it is crucial that {\em only} the prices of items in the neighborhood of the critical set are increased, at every vector of prices through the course of the algorithm. That is, the prices of items in $N(A)$ can only be raised as long as we can be sure that the critical set and its neighborhood of the resulting demand graph does not change.
Due to the generality of the utility functions we consider, it is possible that the neighborhood (which is the demand set) of a buyer $i\in A$ changes ``non-monotonically" if the prices of items in $N(A)$ are not increased carefully, in the sense that items can drop out and then return again to $i$'s neighborhood--this can potentially lead to an exponential runtime despite preserving the condition on critical sets required for correctness. The subroutine \increaseprice\ uses the consistency property of the utility functions to increase prices to the maximum extent possible in such a way that an edge $(i,j)$, for any $i\in A$, vanishes from the demand graph {\em if and only if} $u_{ij}(p_j) \leq 0$. That is, item $j$ drops out of $i$'s demand set only if its price reaches to the threshold where $i$ cannot obtain positive utility from $j$. Such an edge $(i,j)$ that is removed from the demand graph never appears again through the remainder of the algorithm; this is crucial to the strongly polynomial runtime.


The second challenge is to deal with the discontinuity introduced by the budget constraint: For any edge $(i,j)$ in the demand graph, while buyer $i$'s utility from buying an item passes {\em continuously} through the value constraint $v_{ij}\triangleq p_j$ where $u_{ij}(p_j)=0$ (i.e., $i$ is indifferent between buying and not buying the item), there is a discontinuity at budget constraint $p_j=b_{ij}$ where $i$ is {\em not indifferent} between these two actions.
While the budget constraint has the same property as the value constraint that a buyer obtains negative utility at any higher price, our algorithm needs to account for a change in edge structure in the demand graph differently:
At price $p_j=b_{ij}$, buyer $i$ still strictly prefers to buy the item, so the price of this item would have to be strictly higher in any equilibrium.
The algorithm accounts for the discontinuity introduced by the budget using a careful marking process that tags such items whose prices need to be increased later to ensure envy-freeness.
The set of marked items in the final output of the algorithm has a nice property relating the two solution concepts weakly and strongly stable matchings (Proposition~\ref{prop-alg-weakly-stable}).

Finally, in Section~\ref{section-min-mechanism}, we consider the minimum equilibrium as a mechanism, and show that truth-bidding is a Nash equilibrium if there is a competitive equilibrium.

\subsection{Related Work}\label{intro-related}


The classic paper of Demange, Gale and Sotomayor~\cite{DGS} gives two auction-based processes that converge to a minimum equilibrium for the Shapley-Shubik model~\cite{SS}, introducing the idea of increasing prices of over-demanded items to derive a minimum equilibrium.
The key differences, in addition to the fact our algorithm is strongly polynomial time, are that we need to account for the discontinuity introduced by the budget constraint which means an equilibrium need not exist in our model, and that the more general utility functions we allow requires a more involved price increment process.
We note that the equilibrium existence results in ~\cite{CK,DG,Quinnzi} for general utility functions $u_{ij}(\cdot)$ are non-algorithmic, and also require that the utility functions must be continuous everywhere on $\mathbb{R}$ which rules out budget constraints.

A number of recent papers study the assignment model with budget constraints. Aggarwal et al.~\cite{www09} initiate the study of sponsored search advertising with budget-constrained bidders in the assignment model with quasi-linear utilities and solve the problem of finding a buyer-optimal stable matching. They also show that this simple addition of a budget constraint to quasi-linear utilities is surprisingly powerful and can be used to model a number of different types of buyers, including buyers who do not know the precise values for items but only have a preference ranking amongst them. While our model is similar to~\cite{www09} except for our generalization to utility functions $u_{ij}(\cdot)$, the solution concept used in~\cite{www09} turns out to be that of {\em weak stability} (see Section~\ref{section-stability}). In contrast, as Theorem~\ref{theorem-stability} shows, competitive equilibrium, which is the solution concept we focus on, corresponds to {\em strong stability}, leading to a completely different algorithmic problem (the minimum strongly stable matching, even when it exists, can be quite different from the buyer-optimal weakly stable matching, as we show in Example~\ref{example-no-min-eq-iii} and~\ref{example-weak-stable} in Appendix~\ref{appendix-example}).

A number of other recent papers study the special case of quasi-linear utilities with budget constraints in the assignment model. \cite{stacs10} improves the result in~\cite{www09} to compute a buyer-optimal {\em weakly} stable matching without making an assumption on the market necessary in \cite{www09}. In the same model~\cite{kempe09} asks whether there exists a price vector that can support a {\em given} allocation in an envy-free fashion, and how to compute the minimal and maximal such price vector. Of course, this is different from the question of whether there exists {\em any} allocation and price vector that is envy-free and, in addition, clears the market. In the same model again~\cite{LY} observes that a competitive equilibrium need not exist in the presence of budgets, and then focuses on the problem of suitably modifying the notion of competitive equilibrium to a `rationed equilibrium', which is guaranteed to always exist despite the budget constraint, and designing an auction that returns such an equilibrium.
The solution concept of rationed equilibrium differs from a competitive equilibrium, even in instances where competitive equilibria do exist, as shown by Example~\ref{example-LY} in Appendix~\ref{appendix-example}.

Most recently, Ashlagi et al.~\cite{lavi} study a specialized version of the budget-constrained quasi-linear utilities model in the context of sponsored search, where buyer $i$'s value for an item (slot) $j$ with click-through rate $r_j$ is of the form $v_{ij} = v_i r_j$, and budgets are item-independent $b_{ij} = b_i$. It gives a GSP-like auction which converges to a Pareto efficient envy-free outcome when buyer types $(v_i, b_i)$ are distinct.
This work differs from ours in two ways. First, motivated by the practical need to generalize the GSP auction, \cite{lavi} works under assumptions that guarantee that their auction always returns a nontrivial assignment; in contrast, we ask the theoretical question of whether or not there exists a competitive equilibrium for a given instance, and how to find a minimum one if it does.


Second, the bidder model in \cite{lavi} is subtly different from ours: This difference, while seemingly insignificant, does lead to {\em different outcomes} for an identical input, so that our algorithm cannot be viewed as a generalization of theirs to arbitrary utility functions. For technical convenience, \cite{lavi} assumes that a buyer derives negative utility for prices $p_j \geq b_{ij}$, whereas our bidder model, in keeping with the models introduced in~\cite{www09,LY} (as well as the classical market equilibrium and auction literature~\cite{milgrom}), assumes that a buyer can pay a price up to and including his budget, \ie, he derives negative utility only for prices $p_j > b_{ij}$. In Example~\ref{example-lavi} in Appendix~\ref{appendix-example}, we give instances where the outcomes returned by the auction of \cite{lavi} are {\em not competitive equilibria} in our model (and in fact can differ from a competitive equilibrium in a significant way) --- therefore, the auction in \cite{lavi} does not solve the problem of finding a competitive equilibrium for a special case of our model.


\subsection{Computational Market Equilibrium Paradigm}\label{intro-market}
Our approach to the matching market problem follows the computational competitive market equilibrium paradigm, which has received extensive attention in the theoretical computer science literature.
Competitive market equilibrium is a vital concept in economics and has long been established as a standard benchmark of efficiency and fairness in the analysis of markets~\cite{walras,AD}.
Computational issues in this context have also been extensively studied in computational economics~\cite{Scarf,BS}.
While both the existence and computation issues are closely related to the fixed point problem~\cite{Brouwer}, the problem's computational complexity was identified to be at least as hard as PPAD~\cite{ppad} in terms of the aggregated demand functions~\cite{ppad} through Uzawa's reduction~\cite{Uzawa}, for Leontief economies~\cite{CSVY} through a connection to two-player Nash equilibrium~\cite{DGP,CD}, for Arrow-Debreu markets with Arrow-Debreu utility functions~\cite{DD}, as well as for additively separable utilities~\cite{xi1,VY}.

Algorithmic issues have also attracted much interest recently in the Arrow-Debreu model~\cite{DPS,jain,YYY}, especially in the Fisher market setting~\cite{market1,BS}. Most known polynomial time algorithms to compute a market equilibrium, e.g.~\cite{market1,jain}, rely crucially on the fact that items are divisible (i.e. real variables). Our problem has one important difference --- items are indivisible, \ie, sold completely to at most one buyer. With integral variables, it is known that the problem is in general computationally difficult and polynomial time algorithms are known only for very limited special cases~\cite{DPS}; also, a solution is not guaranteed to exist.
Introducing extra budget constraints on top of the integral variables adds another dimension to the problem, since these render the utility functions non-smooth. These two additional constraints make the traditional primal-dual based approach inapplicable to our problem, and the development of a polynomial time solution poses a combinatorial challenge.

While a competitive equilibrium need not exist in our model, our strongly polynomial time algorithm determines whether or not there is an equilibrium, and computes one in case it exists. In this sense, our work contributes to the successful computational competitive market equilibrium paradigm with new positive algorithmic solutions, albeit for a special (but very important) class of problems, and opens up the possibility of a feasible computational equilibrium pricing model in practical markets.

\section{Model and Competitive Equilibrium}\label{section-pre}

We have a market with $n$ unit-demand buyers, and $m$ indivisible items. Unit demand means that each buyer wants at most one item and indivisible means that each item can be sold to at most one buyer. We will denote buyers by $i$ and items by $j$ throughout.

Buyers' preferences over items are described using utility functions $u_{ij}: \mathbb{R}^+\cup\{0\}\rightarrow \mathbb{R}$, that specify the utility of a buyer for an item as a function of its price: $u_{ij}(p_j)$  is the utility of buyer $i$ for item $j$ when the price of item $j$ is $p_j$. Higher utility items are more preferable; we say that buyer $i$ {\em (strictly) prefers} $j$ to $j'$ if $u_{ij}(p_j)>u_{ij'}(p_{j'})$, is {\em indifferent} between $j$ and $j'$ if $u_{ij}(p_j)=u_{ij'}(p_{j'})$, and {\em weakly prefers} $j$ to $j'$ if $u_{ij}(p_j)\ge u_{ij'}(p_{j'})$.
In particular, a utility of $0$, $u_{ij}(p_j) =0$, means that $i$ is indifferent between buying item $j$ at price $p_j$ and not buying anything at all; a negative utility $u_{ij}(p_j) <0$ means the buyer strictly prefers to not buy the item at price $p_j$.

For each buyer-item pair $(i,j)$, we assume that there is a maximum price $b_{ij} \in [0, \infty)$ that $i$ is able to pay for $j$ (set $b_{ij}$ to be infinity if there is no such upper bound); we will call $b_{ij}$ the {\em budget} specified by buyer $i$ for $j$. We set $u_{ij}(p_j) = -1$ for $p_j > b_{ij}$ (here the value $-1$ can be replaced by any negative number). For simplicity, we also assume that $u_{ij}(p_j)\ge -1$ for any $p_j\ge 0$; this is without loss of generality since negative values are not of interest for any buyer.
We will also assume that there are $m$ dummy buyers each with budget zero and utility zero for each item $j$ when $p_j=0$, i.e., $b_{ij} = 0$ and $u_{ij}(0)=0$ (note that for such buyers, $u_{ij}(p_j)=-1$ when $p_j>0$). This assumption is without loss of generality, and implies that the number of items is always less than or equal to the number of buyers, i.e., $m\le n$.

The utility functions permitted by our model are quite general, and allow modeling a fairly large class of marketplaces:
\begin{itemize}
\item Marketplaces with buyers who have quasi-linear utilities and budgets, $u_{ij}(p_j) = v_{ij}- p_j$ for $p_j \leq b_{ij}$, and negative utility for $p_j > b_{ij}$, where $v_{ij}$ is the value of buyer $i$ for item $j$ and $b_{ij}$ is the corresponding budget. Note that such valuations cannot be captured without the budget constraint --- for example, a buyer might have the same payment limit for all items but different valuations for each of them, all larger than the budget, so the budget constraint is a nontrivial extension to the model.

\item Marketplaces with return on investment (ROI) based buyers with budget constraints, \ie, buyers who want to maximize the ratio $u_{ij}(p_j) =t_{ij}/p_j$ subject to a limit on their payment, where $t_{ij}$ is the value that $i$ has for item $j$ (set $u_{ij}(0)=\infty$ if $t_{ij}>0$; and $u_{ij}(0)=0$ if $t_{ij}=0$). Note that such preferences cannot be written as quasi-linear preferences $v_{ij} - p_j$ (with budgets) for any choice of $v_{ij}$,\footnote{To see why, consider a single buyer with two items $j_1,j_2$ and utility functions $1/p_1$, $2/p_2$ respectively, and budget infinity for both items. This buyer prefers item $j_1$ for all prices $(p_1, p_2)$ such that $2p_1 \leq p_2$. If there are values $v_1, v_2$ (of course $v_1<v_2$) for which these preferences can be rewritten as quasi-linear utilities, item $j_1$ is preferred for all price pairs satisfying $p_2 - p_1 \geq v_2-v_1 = \Delta$. But for any choice of $\Delta$, at prices $(\Delta/3, \Delta)$ the buyer strictly prefers $j_2$ in the quasi-linear utility model but $j_1$ in the ROI model, whereas at prices $(3\Delta, 5\Delta)$ he strictly prefers $j_1$ in the quasi-linear utility model but $j_2$ in the ROI model. Thus the two models cannot be equivalent.}
    (although they can be rewritten as $M - x_{ij}p_j$ for adequately large $M$).

    We point out that we can also model buyers who know only the {\em relative values} $t_{ij}/t_{i1}$ of items for $j = 2, \ldots, m$ ($t_{i1}$ is $i$'s value for the first item), and need not know exactly the magnitudes of their values $t_{ij}$ (set $u_{ij}(p_j) = x_{ij}/p_j$, where $x_{ij} = t_{ij}/t_{i1}$).

\item Marketplaces where buyers who only know their preferences for any given prices, but {\em not their values}: Buyers who can only {\em rank} items in order of preference, and have budget constraints for each item. For example, a buyer who prefers item $j_1$ over all other items as long as its price is less than or equal to $b_{ij_1}$, else prefers item $j_2$ as long as price is less than or equal to $b_{ij_2}$, and so on (set $u_{ij}(p_j) = M_{ij} - p_j$ for adequately large values of $M_{ij}$ that ensures the utility values do not intersect when $p_j \leq b_{ij}$). A simple special case is a buyer who has a fixed preference ranking over items and a single budget constraint.
\end{itemize}

Given an instance of the problem, i.e., a set of $n \cdot m$ utility functions $u_{ij}(\cdot)$, the output of the market is a tuple $(\mathbf{x},\mathbf{p})$, where
\begin{itemize}
\item $\mathbf{x}=(x_1,\ldots,x_n)$ is an \textit{allocation} vector, where $x_i$ is the item that $i$ wins. If $i$ does not win any items, denote $x_i=\emptyset$. Note that different buyers must win different items, i.e., $x_i\neq x_{i'}$ for any $i\neq i'$ if $x_i,x_{i'}\neq \emptyset$.
\item $\mathbf{p}=(p_1,\ldots,p_m)\ge 0$ is a \textit{price} vector, where $p_j$ is the price charged for item $j$.
\end{itemize}

Given an output $(\mathbf{x},\mathbf{p})$, if $x_i=j$ (i.e., $i$ wins item $j$), the \textit{utility} that $i$ receives is $u_{ij}(p_j)$. If $x_i=\emptyset$ (i.e., $i$ does not win any item), his utility is defined to be 0 (for simplicity, we denote this by $u_{ix_i}(p_{x_i})=0$). We consider the following solution concept in this paper.

\begin{defi}[Competitive equilibrium]
We say a tuple $(\mathbf{x},\mathbf{p})$ is a \textup{competitive equilibrium} if (i) for any item $j$, $p_j=0$ if no one wins $j$ in allocation $\mathbf{x}$, and (ii) for any buyer $i$, the utility of $i$ is maximized by his allocation at the given vector of prices. That is,
\begin{itemize}
\item if $i$ wins item $j$ (i.e., $x_i=j$), then $u_{ij}(p_j) \geq 0$ (this implies immediately that $b_{ij}\ge p_j$); and for every other item $j'$, $u_{ij}(p_j) \geq u_{ij'}(p_{j'})$.
\item if $i$ does not win any item, then for every item $j$, $u_{ij}(p_j) \leq 0$.
\end{itemize}
\end{defi}

The first condition above is a market clearing condition, which says that all unallocated items are priced at 0 (or at some given reserve price). The assumption that there is a dummy buyer for each item allows us to assume, without loss of generality, that all items are allocated in all equilibria. The second is a {\em fairness}, or {\em envy-freeness} condition, which says that each buyer is allocated an item that maximizes his utility at these prices (note that if an item is priced above the buyer's budget for that item, he has negative utility and therefore does not want the item). That is, given the budget constraints, if $i$ wins item $j$, then $i$ cannot obtain higher utility from any other item; and if $i$ does not win any item, then $i$ cannot obtain a positive utility from any item at these prices.

\section{Stability}\label{section-stability}
Competitive equilibria are closely related to the concept of {\em stability}, where no pair of agents can mutually benefit by deviating from their current assignment. Observe that one way to interpret the price $p_j$ of an item $j$ is that it is the payoff received by the `seller' of item $j$ if the item is sold (the payoff is $0$ if the item is not sold); all sellers prefer higher payoffs. We can define the following two notions of stability for a two-sided market with transferable utilities:

\begin{defi}[Weak and strong stability]
Given an individually rational allocation $\mathbf{x}$ and an associated payoff vector $\mathbf{p}$, where no buyer derives negative utility and $p_j\ge 0$ is the payoff to seller $j$, we say the tuple $(\mathbf{x},\mathbf{p})$ is
\begin{itemize}
\item {\em weakly stable} if there is no blocking pair $(i,j)$, $j \neq x_i$, such that $u_{ij}(p'_j) > u_{ix_i}(p_{x_i})$ for some $p'_j > p_j$;
\item {\em strongly stable} if there is no blocking pair $(i,j)$, $j \neq x_i$, such that (i) $u_{ij}(p'_j) \ge u_{ix_i}(p_{x_i})$ for some (ii) $p'_j \geq p_j$, and at least one of inequalities (i) and (ii) is strict.
\end{itemize}
\end{defi}
That is, to block a weakly stable matching, {\em both} sides of the blocking pair $(i,j)$ must strictly benefit relative to their current allocation, \ie, buyer $i$'s utility must strictly increase and seller $j$ must be able to receive a strictly higher payoff by deviating. For strong stability, however, only one side in a blocking pair needs to strictly benefit, while the other side need only weakly benefit from the deviation\footnote{Note that when $u_{ij}(\cdot)$ is strictly decreasing, if there is a pair $(i,j)$ having strict inequality (ii) in the definition for strong stability, this pair will have strict inequality (i) as well, so it is enough to check (i) to decide strong stability.}.

When the utility functions $u_{ij}(\cdot)$ are strictly decreasing and continuous everywhere as in \cite{CK,DG}, the conditions for weak and strong stability turn out to be identical: if there is a pair $(i,j)$ where $i$'s utility from $j$ is strictly larger than from $x_i$, there must exist a price $p'_j = p_j + \epsilon > p_j$ at which $i$ still strictly prefers $j$ to $x_i$. That is, if $i$ strictly prefers $j$, it is always possible for $i$ and $j$ to deviate in such a way that both $i$  and $j$ strictly benefit from the deviation. So weakly and strongly stable matchings are identical, and there is a unique notion of stability in the original matching model.

However, with the budget constraint, this is no longer the case: suppose there is a pair $(i,j)$ such that $u_{ij}(p_j) > v_{ix_i}(p_{x_i}) \geq 0$ as before. Depending on whether $p_j < b_{ij}$ or $p_j = b_{ij}$, there may or may not exist a strictly profitable deviation for {\em both} $i$ and $j$: in the first case, there exists a $p'_j > p_j$ with $u_{ij}(p'_j) > u_{ix_i}(p_{x_i})$; but in the second case, there is no $p'_j > p_j$ for $i$ to continue to prefer $j$ over $x_i$, \ie, there is no strictly profitable deviation for $j$. That is, with the addition of the budget constraint, the two notions of ``weak'' and ``strong'' stability are no longer equivalent as in the original assignment model: the non-equivalence is precisely because a buyer's utility can go from strictly positive to strictly negative without passing through $0$ at the point of discontinuity at $b_{ij}$. We point out the analogy with the situation in the Gale-Shapley marriage model~\cite{gs}, when ties are introduced into preference lists, the two notions of stability no longer coincide~\cite{irving}.

Since there are two distinct notions of stability in the matching model with budgets, at most one of these can be the same as a competitive equilibrium.
The following claim shows that the solution concept of a competitive equilibrium coincides exactly with that of strong stability\footnote{The reason that notion of stability used in~\cite{www09} corresponds to weak stability is because the inequality (4) in~\cite{www09} is not strict, which translates exactly to being able to strictly increase the price for item $j$ when $j$ belongs to a blocking pair.}.

\begin{theorem}\label{theorem-stability}
Suppose that the utility functions $u_{ij}(\cdot)$ are strictly decreasing on domain $[0,b_{ij}]$. Then, $(\mathbf{x},\mathbf{p})$ is a competitive equilibrium if and only if it is also strongly stable.
\end{theorem}
\begin{proof}
If $(\mathbf{x},\mathbf{p})$ is a competitive equilibrium, since all unallocated items are prices at 0, each $p_j$ is precisely the payoff that the corresponding seller $j$ receives. Since $u_{ix_i}(p_{x_i}) \geq 0$ and $p_j\ge 0$ for any $i$ and $j$, $(\mathbf{x},\mathbf{p})$ is individually rational. By the definition of competitive equilibrium, we have $u_{ix_i}(p_{x_i}) \geq u_{ij}(p_j)$ for any $j \neq x_i$. If $p_j=b_{ij}$, clearly $i$ and $j$ cannot be a blocking pair. If $p_j<b_{ij}$, since the utility functions are strictly decreasing, $u_{ij}(p'_j) < u_{ij}(p_j)$ for all $p'_j > p_j$; thus $i$ and $j$ are not a blocking pair as well. Hence, there exists no strongly blocking pair $(i,j)$ and $(\mathbf{x},\mathbf{p})$ is strongly stable.

Conversely, if $(\mathbf{x},\mathbf{p})$ is strongly stable, then $u_{ix_i}(p_{x_i})\ge 0$ and $p_j\ge 0$ for any $i$ and $j$ by individual rationality. Consider each $p_j$ as the price of item $j$. If item $j$ is not sold to any buyer, the payoff that seller $j$ receives is 0 and thus $p_j=0$; hence the market clearing condition holds. For any buyer $i$, if there is $j\neq x_i$ such that $u_{ij}(p_j)>u_{ix_i}(p_{x_i})$, then $(i,j)$ would be a strongly blocking pair since $i$ obtains more utility and seller $j$ gets the same amount of payoff. Hence, the envy-freeness condition also holds, which implies that $(\mathbf{x},\mathbf{p})$ is a competitive equilibrium.
\end{proof}

\section{Computing a Minimum Competitive Equilibrium}\label{section-min}

In this section, we present a strongly polynomial time algorithm to determine if a competitive equilibrium exists, and find a minimum one if it does, for a class of utility functions $u_{ij}(\cdot)$ that satisfy the conditions below; the first two are identical to those required in~\cite{CK,DG}.
We will refer to a set of $mn$ utility functions $u_{ij}(\cdot)$ that satisfy these properties as {\em consistent} utility functions.

\begin{enumerate}
\item Continuity. Each function $u_{ij}(\cdot)$ is continuous on $[0,b_{ij}]$.\footnote{Note that since $u_{ij}(p_j)=-1$ when $p_j>b_{ij}$, the utility function $u_{ij}(\cdot)$ might not be continuous on the whole domain $\mathbb{R}^+\cup \{0\}$. Actually, we really only require $u_{ij}(\cdot)$ to be continuous where it is non-negative (note that $u_{ij}(p_j)$ can be negative for $p_j < b_{ij}$, for instance, with quasi-linear utilities where $v_{ij} < b_{ij}$), but requiring the property on $[0,b_{ij}]$ is without loss of generality since negative values are not of interest.}

\item Monotonicity. Each function $u_{ij}(\cdot)$ is strictly decreasing on $[0,b_{ij}]$.

    Since the $u_{ij}$ is strictly decreasing, we can define the inverse function $u_{ij}^{-1}(q) = p$ if $u_{ij}(p)=q$  for any $q \in \mathbb{R}$; if there is no such $p$, define $u_{ij}^{-1}(q)=\infty$. Define the threshold {\em value} $v_{ij}=u_{ij}^{-1}(0)$, which (if $v_{ij}\neq\infty$) is the price at which buyer $i$ becomes indifferent between buying $j$ and not buying anything. For quasi-linear utilities, $v_{ij}$ is exactly the value of buyer $i$ for item $j$; however for ROI-based buyers, this value $v_{ij}$ is $\infty$ (and $v_{ij}=0$ if $t_{ij}=0$; for such case, to guarantee monotonicity, we can set $b_{ij}=0$).

\item Consistency. The consistency condition is the one that relates utility functions $u_{ij}(\cdot)$ across buyers and items by transitive paths. We say that a path $P=(j_1,i_1,j_2,\ldots,i_{\ell-1},j_\ell)$ is {\em transitive} with respect to price vector $\mathbf{p}$ if $u_{i_kj_k}(p_{j_k}) = u_{i_kj_{k+1}}(p_{j_{k+1}})\ge 0$ for $k=1,\ldots,\ell-1$ (buyers and items can be repeated). That is, $P$ is such that each buyer $i_k$ gets equal utility from its two neighboring items $j_{k}$ and $j_{k+1}$ in the path. The consistency property relates $u_{ij}(\cdot)$ using such transitive paths, defined formally as follows.

    Suppose a path $P$ is transitive with respect to $\mathbf{p}$, as well as with respect to another price vector $\mathbf{q} > \mathbf{p}$, where each price $q_{j_k}$ is within the value and budget constraints of its neighbors on path $P$. Then any buyer $i$ in the market (not only those on path $P$) who weakly prefers $j_1$ to $j_\ell$ in $\mathbf{p}$ continues to weakly prefer $j_1$ to $j_\ell$ in $\mathbf{q}$ when $q_{j_1}$ is within the value and budget constraints of $i$. That is, if $q_{j_1}\le \min\{v_{ij_1},b_{ij_1}\}$, then $u_{ij_1}(p_{j_1})\ge u_{ij_\ell}(p_{j_\ell})$ implies $u_{ij_1}(q_{j_1})\ge u_{ij_\ell}(q_{j_\ell})$.\footnote{Note that if $u_{ij_1}(p_{j_1})= u_{ij_\ell}(p_{j_\ell})$, we can use the two inequalities (the other one is by switching $j_1$ and $j_\ell$) to conclude that $u_{ij_1}(q_{j_1})= u_{ij_\ell}(q_{j_\ell})$, given value and budget constraints on $j_1$ and $j_\ell$.} In other words, when prices are increased from $\mathbf{p}$ to $\mathbf{q}$ while maintaining the transitivity of path $P$, all buyers continue to have the same preference over items in the path in both price vectors $\mathbf{p}$ and $\mathbf{q}$ (subject to their value and budget constraints).
\end{enumerate}

While the consistency condition might appear to be strong, it is easy to verify that it holds for each of the buyer models discussed in the previous section. For example, for quasi-linear buyers with budgets, the only way to maintain the transitivity of a path starting from a price vector $\mathbf{p}$ is to increase all prices by the same increment $\epsilon$, i.e., $p_j\leftarrow p_j+\epsilon$. Since the price increment is identical for all items, all buyers retain their preference ordering across items. For ROI-based buyers, transitivity and consistency hold when $p_j\leftarrow (1+\epsilon)p_j$. We note that the consistency condition is not about the (existence of) price vector $\mathbf{q}$ itself, but rather, the relative preference ordering of buyers over items in the two price vectors $\mathbf{p}$ and $\mathbf{q}$ (the existence of such $\mathbf{q}$ can be shown easily using the continuity and monotonicity of the $u_{ij}(\cdot)$ as long as the relevant value and budget constraints are not tight at $\mathbf{p}$). The consistency property gives us a way to increase prices efficiently (see Section~\ref{section-price-increase}).

In what follows, we will assume that all utility functions $u_{ij}(\cdot)$, for $i=1, \ldots, n$ and $j = 1, \ldots,m$, of a given instance satisfy the above conditions. Naturally, the functions $u_{ij}$ and their inverses $u^{-1}_{ij}$ must be polynomial time computable as well; this is an implicit assumption in all our results.
Our main result is the following.

\begin{theorem}\label{theorem-main}
Suppose we are given an instance of the assignment model with consistent utility functions $u_{ij}$. Then, if a competitive equilibrium exists, a minimum competitive equilibrium exists as well; further, the problem of deciding whether or not an equilibrium exists, and computing a minimum one, can be solved in strongly polynomial time.
\end{theorem}

We must first clarify what we mean by a minimum equilibrium --- There are three reasons a minimum equilibrium may not exist. (i) First, there may exist no equilibrium at all for the given instance, due to the budget constraint, as shown by Example~\ref{example-no-eq}. (ii) Second, equilibria may exist, but the associated price vectors may not be {\em comparable}, so that a minimum equilibrium does not exist. (Indeed, Example~\ref{example-no-min-eq} in Appendix~\ref{appendix-example} shows that if the utility functions are not strictly decreasing, a minimum equilibrium does not exist because the equilibrium price vectors are incomparable.) (iii) Finally, since we deal with real number prices rather than restricting to integer prices, the set of equilibrium prices for an item need not contain its infimum; strictly speaking, therefore, a minimum equilibrium price need not exist even though equilibria might exist and are comparable.

It will turn out that for utility functions in our model (\ie, satisfying the above three conditions), (ii) never happens. However, we will want to distinguish between instances of type (i) which have no competitive equilibrium at all, versus those of type (iii) where an equilibrium exists, but the set of equilibrium prices does not contain the infimum.

We therefore define a `$p+$' notation to deal with such instances: a minimum equilibrium with price $p_j+$ for an item $j$ means that there is no equilibrium with that item priced at $p_j$ or less, but there does exist an equilibrium with price $p_j + \epsilon_j$, where $\epsilon_j>0$ is an arbitrarily small precision. The term `$p+$' does not really refer to any particular price, but is just our notation for the concept of a real number that can be arbitrarily close to $p$ (from the right). In particular, for any given real number $p'>p$, we have $p<p+<p'$. We use this notation to formally define a minimum competitive equilibrium as follows.

\begin{defi}[Minimum competitive equilibrium]\label{def-min-eq}
Define the {\em infimum price vector} $\mathbf{p}$ where $p_j$ is the infimum of all equilibrium prices for item $j$, and let $T=\{j~|~\textup{there is an equilibrium where $j$ is priced at $p_j$}\}$. Suppose there is an allocation vector $\mathbf{x}^*$ such that for every $\epsilon > 0$, there exist $0\le \epsilon_j \leq \epsilon$ for which $(\mathbf{x}^*,\mathbf{q})$ is a competitive equilibrium, where $q_j=p_j$ if $j\in T$ and $q_j = p_j + \epsilon_j$ otherwise. Then we say $(\mathbf{x}^*,\mathbf{p}^*)$ is a {\em minimum competitive equilibrium} and $\mathbf{p}^*$ is the {\em minimum equilibrium price vector}, where $p^*_j=p_j$ if $j\in T$ and $p^*_j=p_j+$ otherwise.
\end{defi}

Note that $\mathbf{p}^*$ is not really a vector of prices, since it includes some entries of the form `$p+$': the definition above gives a precise way to translate this ``conceptual price vector" into an actual vector of prices.
For example, consider the quasi-linear $v_{ij}-p_j$ utility model with budgets, with three buyers $i_1,i_2,i_3$ and two items $j_1,j_2$. Every buyer-item pair $(i,j)$ has the same value  $v_{ij}=10$; and $b_{i_1j_1}=b_{i_1j_2}=b_{i_2j_1}=b_{i_2j_2}=10$ and $b_{i_3j_1}=b_{i_3j_2}=2$. Then in the above definition, we have $\mathbf{p}=(2,2)$ and $\mathbf{p}^*=(2+,2+)$. Hence, $(\mathbf{x}^*,\mathbf{p}^*)$, where $x^*_1=j_1$ and $x^*_2=j_2$, is a minimum competitive equilibrium --- for any small $\epsilon>0$, $(p_1+\epsilon,p_2+\epsilon)$ is an equilibrium price vector supporting $\mathbf{x}^*$.

The infimum price vector $\mathbf{p}$ and minimum equilibrium price vector $\mathbf{p}^*$ are both uniquely defined for any given instance (if no equilibrium exists at all, they can be defined as $\infty$). In general, there may be no equilibrium price vector associated with $\mathbf{p}^*$, as Example~\ref{example-no-min-eq} in Appendix~\ref{appendix-example} shows. Theorem~\ref{theorem-main}, however, implies that when the utility functions are consistent, there must exist an equilibrium price vector associated with $\mathbf{p}^*$ whenever $\mathbf{p}^*\neq \infty$, i.e., the instance has a minimum equilibrium price vector, which is $\mathbf{p}^*$.

To prove Theorem~\ref{theorem-main}, we will first begin with some essential preliminaries in Section~\ref{section-alg-pre}, after which we describe the price increment process and its properties in Section~\ref{section-price-increase}. We finally present the algorithm \algmin\ and outline its proof of correctness in Section~\ref{section-alg-main}. All proofs can be found in the Appendix.

\subsection{Preliminaries}\label{section-alg-pre}

\paragraph{Dynamic Demand Graph $G$ and $G^+$.}
Given a price vector $\mathbf{p}=(p_1,\ldots,p_m)$, define its associated \textit{demand bipartite graph} to be $G(\mathbf{p})=(U,V;E)$, where $U$ corresponds to the set of buyers and $V$ corresponds to the set of items, and $(i,j)\in E$ if $u_{ij}(p_j) > 0$ and $u_{ij}(p_j) \geq u_{ij'}(p_{j'})$ for any $j'\in V$. That is, for the given price vector $\mathbf{p}$, $N(i)$ gives the \textit{demand set} of buyer $i$, i.e., items that bring maximal, strictly positive, utility to buyer $i$. In the algorithm, when prices change, the demand set of every buyer will be updated accordingly, as also the edge set $E$.

Note that in graph $G$, there may be isolated buyers in $U$ which are priced out of the graph because their utility becomes non-positive for every item (\ie, for each item $j$,  $p_j \geq v_{ij}$ or $p_j > b_{ij}$). That is, $i\in U$ is isolated if for every $j\in V$ $u_{ij}(p_j) \leq 0$.
At the same time, items in $V$ may also be isolated since no buyer can get a maximal positive utility from them.
We denote $U^+=\{i\in U~|~N(i)\neq \emptyset\}$ and $V^+=\{j\in V~|~N(j)\neq \emptyset\}$ to be the set of non-isolated buyers and items in graph $G$, respectively, and define $G^+(\mathbf{p})=(U^+, V^+; E)$. Clearly $G^+$ is a subgraph of $G$: they have the same edge set $E$, and a vertex of $G$ is in $G^+$ only if it has a non-empty neighbor set. Note that as the algorithm develops, $U$ and $V$ are fixed and $E$ is the only dynamic set in $G$; whereas in $G^+$, both $U^+$ and $V^+$ are dynamic as well. The critical set of $G^+$ will play a central role in the algorithm.

\paragraph{Critical Set.}
The notion of the {\em critical set} in a bipartite graph is used centrally by our algorithm to identify over-demanded items in $G^+(p)$ that ``block" a competitive equilibrium at price $p$.
Given a bipartite graph $G=(U,V;E)$, the \textit{deficiency} of a subset $A\subseteq U$ is defined to be $\delta(A) = |A|-|N(A)|$, where $N(A)\subseteq V$ is the set of neighbors of $A$. For simplicity, we denote $N(\{i\})$ by $N(i)$ and $N(\{j\})$ by $N(j)$. The deficiency of graph $G$ is defined to be $\delta(G)=\max_{A\subseteq U}|A|-|N(A)|$, the maximum deficiency taken over all subsets of $U$. Since the deficiency of an empty set is 0, $\delta(G)\ge 0$. (Note that symmetrically, the deficiency $\delta(G)$ can be defined in terms of vertices in $V$ as well.)

A maximally deficient set is a subset $A\subseteq U$ such that $\delta(A)=\delta(G)$. A subset $A\subseteq U$ is called \textit{critical} if it is maximally deficient and contains no maximally deficient proper subset. Note that if $A_1$ and $A_2$ are maximally deficient, then so does $A_1\cap A_2$. Hence, there is a unique critical set~\cite{liu,irving}.
If graph $G$ has no non-empty critical set, we have $|A|\ge |N(A)|$ for any $A\subseteq U$; thus by Hall's theorem~\cite{graphbook}, there is a maximum matching of size $|U|$ in $G$.  Irving~\cite{irving} showed a simple polynomial time algorithm to find the critical set.

The following theorem about critical sets, proved in Appendix~\ref{appendix-theorem-critical}, will be needed for the algorithm.

\begin{theorem}\label{them-critical}
Given a bipartite graph $G=(U,V;E)$, let $A\subseteq U$ be the critical set of $U$. Then the following two claims hold:
\begin{itemize}
\item If we add some edges between $A$ and $N(A)$, $A$ remains the critical set of the resulting graph.
\item If we delete some edges between $U\setminus A$ and $N(A)$, $A$ remains the critical set of the resulting graph.
\end{itemize}
\end{theorem}

\subsection{Increasing Prices: Subroutine \increaseprice}\label{section-price-increase}

In this subsection, we describe the subroutine \increaseprice\ used by the main algorithm \algmin\ to increase prices of over-demanded items.
\increaseprice\ operates on a vector of prices $\mathbf{p}$ and a subset of buyers $S$, and raises the prices of items that belong to the neighborhood $N(S)$ of $S$ in $G^+(\mathbf{p})$ in a manner that preserves transitivity of all paths between $S$ and $N(S)$ at price $\mathbf{p}$. The output returned by \increaseprice\ is the smallest price vector $\mathbf{q}$ at which either an item $j \notin N(S)$ is added to the demand set of a buyer $i \in S$, or the price of some item $j \in N(i)$ reaches either the value $v_{ij}$ or the budget $b_{ij}$ of a buyer $i \in S$.

The main algorithm will use \increaseprice\ to increase the prices of items in the neighborhood $N(S)$ of the critical set $S$ of buyers in the dynamic demand graph $G^+(\mathbf{p})$. This set of items in $N(S)$ is over-demanded at price vector $\mathbf{p}$ when $S\neq \emptyset$, so there can be no equilibrium at $\mathbf{p}$ since not all buyers in $S$ can be matched to distinct items in $N(S)$. Ideally, we would like to increase prices of these over-demanded items as much as possible to make the algorithm efficient. However, to ensure that the main algorithm does converge to a (minimum) equilibrium price vector, it is crucial that {\em only} the prices of items in the neighborhood of the critical set are increased at every vector of prices through the course of the main algorithm. The properties of the utility functions, together with Theorem~\ref{them-critical} on the structure of critical sets, allows us to guarantee that there is a way to increase prices such that the structure of the critical set does not change {\em until} the price of an item $j \in N(i)$ reaches either the value $v_{ij}$ or budget $b_{ij}$ of a buyer $i\in S$, or becomes large enough to match the utility of an item outside $N(S)$.
(Note that with arbitrary utility functions and arbitrary price increments, an item $j\in N(i)$ may drop out of the demand set of $i$ before its utility reaches zero because it no longer gives maximal utility, which can also cause a change in the structure of the critical set. However, the consistency property imposed on the utility functions gives a way to increase prices (preserving transitivity) for which this cannot happen, so that it is enough to check for these three conditions. See Appendix~\ref{appendix-alg-increase} for a more detailed discussion.)

In general, the price vectors at which these three events ($p_j = v_{ij}$, $b_{ij}$, or $u_{ij'}^{-1}(p_{j'})$ for $j' \notin N(S)$) occur for different buyer-item pairs $(i,j)$ may also not be comparable.
Once again, the properties of the utility functions allow us to guarantee that this does not happen, as shown in \increaseprice\ and its proof of correctness in Appendix~\ref{appendix-alg-increase}.

The subroutine \increaseprice\ is given below. The connected component $C$ in \increaseprice\ is specified as an input by the main algorithm, and will turn out to be a connected component of the critical set and its neighborhood in $G^+$.

\begin{center}
\small{}\tt{} \fbox{
\parbox{6.4in}{
\footnotesize

\hspace{0.05in} \\[-0.01in] \increaseprice

\begin{enumerate}

\item Let $\mathbf{p}=(p_1,\ldots,p_m)$ denote the current price vector and $G^+(\mathbf{p})$ be the associated demand graph

\item Given a connected component $C=S \cup N(S)$ in $G^+$ 
\item For each edge $(i_0,j_0) \in C$
\begin{enumerate}
\item let $u^{i_0}_{\max} = \max_{j \notin N(S)} u_{i_0j}(p_{j})$ be the maximal utility $i_0$ obtains from items not in $N(S)$

\item let $q^{i_0j_0} = \min\left\{v_{i_0j_0}, b_{i_0j_0}, u^{-1}_{i_0j_0}(u^{i_0}_{\max})\right\}$, where $v_{i_0j_0}=u_{i_0j_0}^{-1}(0)$ and $b_{i_0j_0}$ is the budget

\item define vector $\mathbf{q}^{i_0j_0} = (q^{i_0j_0}_j)_{j \in N(S)}$ as follows:

\begin{itemize}
\item let $q^{i_0j_0}_{j_0}=q^{i_0j_0}$
\item for any pair of edges $(i,j),(i,j')\in C$ with $q^{i_0j_0}_{j'}$ being defined and $q^{i_0j_0}_{j}$ not, let \\ $q^{i_0j_0}_{j}=u_{ij}^{-1}\left(u_{ij'}(q^{i_0j_0}_{j'})\right)$  \ \
    (if $q^{i_0j_0}_{j}>q^{i'j}=\min\left\{v_{i'j},b_{i'j},u^{-1}_{i'j}(u^{i'}_{\max})\right\}$ for any \\
    $i'\in S\cap N(j)$, set $\mathbf{q}^{i_0j_0} = \infty$ and break the local "for" loop of Step~(3))
\end{itemize}
\end{enumerate}

\item Define $\mathbf{q}=(q_j)_{j\in N(S)}$ to be the minimum of the vectors $\mathbf{q}^{i_0j_0}$ for all $(i_0,j_0) \in C$

\item Set $p_j = q_j$ for each $j\in N(S)$ and $p_j=p_j$ for each $j\notin N(S)$.
\end{enumerate}

}}
\end{center}

\setlength{\baselineskip}{.5cm}

Since both $u_{ij}(\cdot)$ and $u_{ij}^{-1}(\cdot)$ can be computed in polynomial time, $u^{i_0}_{\max}$, $q^{i_0j_0}$, and therefore $\mathbf{q}^{i_0j_0}$, can be computed in polynomial time. Thus \increaseprice\ is in strongly polynomial time.

Note that for each $(i_0,j_0) \in C$, there could be different paths that lead to defining the value $q^{i_0j_0}_{j}$ for an item $j$ in Step~(3), and it is not a priori obvious that each of these leads to the same value. However, as the first claim of the following theorem implies, the price vector $\mathbf{q}^{i_0j_0}$ is indeed uniquely defined in Step~(3). It is also not obvious that these price vectors can all be compared; the second claim says that all different $\mathbf{q}^{i_0j_0}$ and $\mathbf{q}^{i'_0j'_0}$ are comparable; thus the minimum price vector $\mathbf{q}$ in Step~(4) is well-defined and satisfies $\mathbf{q}\neq \infty$. The last claim says that all buyers in $S$ continue to weakly prefer their neighbors in $N(S)$ in $G^+(\mathbf{p})$ with respect to price vector $\mathbf{q}$; this property is crucial to the analysis of the main algorithm.

\begin{theorem}\label{theorem-price-increase}
Suppose that the utility functions are consistent. Given initial price vector $\mathbf{p}$ in Step~(1), the following claims hold in \increaseprice:
\begin{itemize}
\item Price vector $\mathbf{q}^{i_0j_0}$ in Step~(3) is well-defined for any edge $(i_0,j_0) \in C$, and $\mathbf{q}^{i_0j_0}\ge \mathbf{p}$: any two alternate ways to define $q^{i_0j_0}_{j}$ in Step~(3) lead to the same value (that is, suppose there are $(i,j),(i,j')$ and $(i',j),(i',j'')$ in $C$, where both $q^{i_0j_0}_{j'}$ and $q^{i_0j_0}_{j''}$ have already been defined. The value $q^{i_0j_0}_{j}$ will be the same irrespective of which of these is used to define it).

\item For any edges $(i_0,j_0),(i'_0,j'_0) \in C$, the vectors $\mathbf{q}^{i_0j_0}, \mathbf{q}^{i'_0j'_0}$ are {\em comparable}; the minimum price vector $\mathbf{q}$ defined in Step~(4) exists and satisfies $q_{j_0}\le q^{i_0j_0}_{j_0}$ for all $(i_0,j_0)\in C$ (this implies that $\mathbf{q}\neq\infty$).

\item For any buyer $i\in S$ and item $j\in N(i)$, where $N(i)$ is the neighborhood of $i$ in $G^+(\mathbf{p})$, $i$ weakly prefers $j$ to all other items with respect to price vector $\mathbf{q}$. That is, $u_{ij}(q_j)\ge u_{ij'}(q_{j'})$ for any $j'\in V$.
\end{itemize}
\end{theorem}


\subsection{Main Algorithm \algmin}\label{section-alg-main}

The algorithm is essentially composed of two parts --- Steps~(1-6) to eliminate all possible critical sets and Steps~(7-9) to determine if there is a feasible assignment to clear the market. Starting with a price of $0$ (or any reserve price vector) for all items, Step~(3) of the algorithm recursively increases prices to eliminate critical sets corresponding to over-demanded items using the subroutine \increaseprice.  The increased price vector returned by \increaseprice\ corresponds to to one of the following three events, at which the critical set $S$ or its neighbor set $N(S)$ of the dynamic demand graph might change (the other two possibilities, edges added between $S$ and $N(S)$ or edges deleted between $U\setminus S$ and $N(S)$, do not change the critical set by Theorem~\ref{them-critical}). In each case, we appropriately update the dynamic graphs $G$ and $G^+$ (Step (3)) as follows and proceed with the new demand graph and its critical set:
\begin{itemize}
\item A buyer becomes indifferent between his neighboring items in $N(S)$ and some item in $V^+ \setminus N(S)$, in which case we simply update $G$ and $G^+$ and return to the price increment process again.
\item The price $p_j$ of $j \in N(S)$ reaches the utility threshold of some neighbor $i$, \ie, $p_j = v_{ij}$ for some edge $(i,j) \in G^+$: We remove all edges incident to the buyer since its maximal utility has dropped to zero (note that at this new price vector, $i$'s utility for {\em all} items in $N(i)$ is zero). Clearly, the buyer will not obtain positive utility afterwards (\ie, permanently priced out of the algorithm) and become an isolated vertex in $G$, and thus, will not belong to $G^+$ any more.
\item The price $p_j$ of $j \in N(S)$ reaches the budget of some neighbor $i$, \ie, $p_j = b_{ij}$ for some edge $(i,j) \in G^+$: The edge $(i,j)$ will be deleted\footnote{Note that after this edge is deleted, new edges may immediately appear between $i$ and other utility-maximizing items $j'$ with $u_{ij'} > 0$ (if such exist) in the dynamic demand graph $G^+$.} (permanently) and a {\em marking} operation will be performed on items in $N(S)$. The marking process is because of the need to ensure that the price of such an item in any output returned by the algorithm is strictly larger than the price at which this edge was deleted, since otherwise the item will be over-demanded and the resulting output cannot be an equilibrium.
    We note that the status of an item, marked or unmarked, remains the same in all subsets of the demand graphs used by the algorithm.
    In addition, the set of marked items in the final stage of the algorithm has an interesting property related to weakly stable matchings (see Proposition~\ref{prop-alg-weakly-stable}).
\end{itemize}

Once the algorithm eliminates all critical sets in $G^+$, it exits Step (3) with some price vector, say $\mathbf{p}^*$. The set of items now contains some marked and some unmarked items; the marked items are those whose price must be strictly larger than $p^*_j$  in any equilibrium whereas the price of unmarked items need not be increased. It is possible that there is a buyer $i\in U^+$ such that some items in $N(i)$ are marked and some are unmarked (e.g., node $i_2$ in Example~\ref{example-alg}): This buyer $i$, who was indifferent between a marked item $j_1$ and an unmarked item $j_2$ at prices $\mathbf{p}^*$, will no longer be indifferent after the prices of marked items are raised, but will strictly prefer $j_2$. In this sense, graph $G^+(\mathbf{p}^*)$ does not correctly reflect the demand sets for buyers and a matching in $G^+(\mathbf{p}^*)$ need not be a competitive equilibrium: the graph must be further processed to ensure that every buyer is genuinely indifferent between his neighbors in the demand graph, \ie, either {\em all} items in his demand set are marked, or {\em all} items in his demand set are unmarked. The reduced graph $G'$ is defined in Step~(4) as a subgraph of $G^+(\mathbf{p}^*)$ containing only edges to unmarked items, and Step (5) deals with nonempty critical sets in $G'$ by marking all items in its neighborhood. Note that if the critical set and its neighborhood includes an edge which is tight on the budget constraint, such an edge must be deleted (since after the price increase due to marking, this edge will lead to negative utility), returning us to Step (3), since we can no longer guarantee that the critical set in $G^+(\mathbf{p}^*)$ is the same as before Step (5) (specifically, that it is empty).

Finally, since prices keep increasing and deleted edges will never appear again, the algorithm arrives at Step (6), where we construct $G^*$, a subgraph of $G'$, with the property that the demand set of every buyer in $U^+$ contains either all marked or all unmarked items, so $G^*$ captures the exact demand relation for all buyers.

Steps~(7-9) of the algorithm determine if there is a feasible assignment. Let $N^*(i)$ denote the neighbor set of $i$ in $G^*$ defined in Step~(6). Since $G^*$ correctly illustrates the demand sets of all buyers, any buyer $i$ with $N^*(i)\neq \emptyset$ obtains his maximal positive utility from items in $N^*(i)$. The construction of $G'$ and $G^*$ is such that since $G'$ has no critical set, $G^*$ does not either (Proposition~\ref{prop-no-critical}), so that all buyers in $U^+$ can be matched to an item in their demand sets (note that the set of buyers in $G^*$ is the same as $U^+$).
The only remaining condition that needs to be satisfied to guarantee a competitive equilibrium is that every item that has a price greater than $0$ (or its reserve price) can be matched to a buyer. We therefore construct graph $H\supseteq G^*$ in Step~(7) by adding those buyers who obtain maximal utility 0 from certain items back into consideration. That is, we add edge $(i,j)$ to graph $H$ if the buyer $i\notin U^+$ derives utility 0 from an unmarked item $j$ and can afford it ($b_{ij}\ge p_j$). (Only unmarked items are considered since if $j$ is marked, $p_j$ will be set to be $p_j+$ in Step~(8.b) and $i$ will obtain a negative utility from $j$.) These edges added in Step~(7) help us to assign as many items as possible. Finally, if there is a maximum matching of $H$ with all items being assigned and all buyers in $U^+$ being matched (to their neighbors in $G^*$), it is returned as an equilibrium allocation; if no such maximum matching exists, Step~(9) reports that there is no competitive equilibrium.

The algorithm is given formally in the next page. Example~\ref{example-alg} below illustrates partial stages of Step (3) and (5) of the algorithm.

\begin{exam}\label{example-alg}
\small
Consider the quasi-linear $v_{ij}-p_j$ utility model with five buyers $i_1,\ldots,i_5$ and three items $j_1,j_2,j_3$. All buyers have the same budgets for all three items, $\infty,190,2,1,1$, respectively. The value vectors for three items are $(1000,100,100)$ and $(200,11,11)$ for buyer $i_1$ and $i_2$, respectively, and $(20,10,10)$ for $i_3,i_4,i_5$. The stages of each run of Step~(3) of the \algmin\ are shown in the following first five figures (where black vertices on the left  denote the critical set and on the right denote marked items, respectively):
\begin{figure}[ht]
\begin{center}
\includegraphics[scale = 0.9]{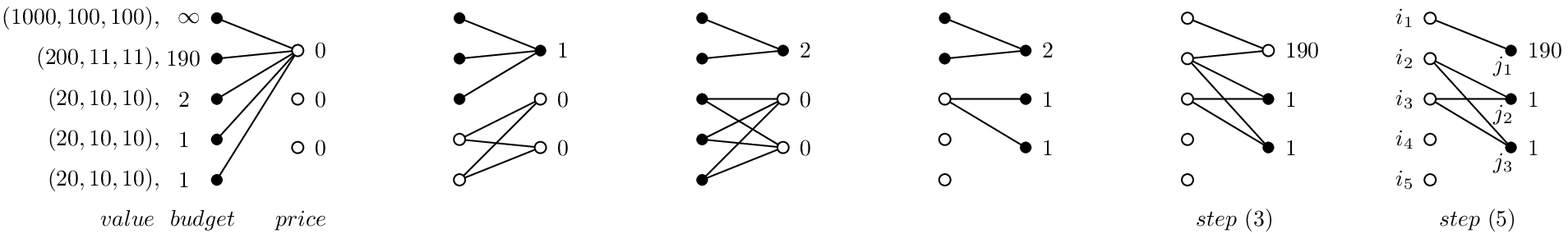}
\end{center}
\end{figure}
\\[-.2in]
\noindent After Step~(3) (the above fifth figure), $i_2$ has three neighbors, unmarked $j_1$ and marked $j_2,j_3$, which means that $i_2$ will strictly prefer $j_1$ to $j_2,j_3$. If we restrict on reduced subgraph $G'$ containing only unmarked item $j_1$, a new critical set arises, which is $\{i_1,i_2\}$. Hence, in Step~(5) the algorithm will set $j_1$ to be marked. When this happens, since the budget of $i_2$ on $j_1$ is tight (i.e., $b_{i_2j_1}=190=p_{j_1}$), we will have to delete edge $(i_2,j_1)$ in Step~(5.b). Eventually the algorithm will return assignment $\{(i_1,j_1),(i_2,j_2),(i_3,j_3)\}$ (or $\{(i_1,j_1),(i_2,j_3),(i_3,j_2)\}$) and minimum equilibrium price vector $(190+,1+,1+)$.
\end{exam}

\begin{center}
\small{}\tt{} \fbox{
\parbox{6.4in}{
\footnotesize

\hspace{0.05in} \\[-0.01in] \algmin

\begin{enumerate}
\setlength{\baselineskip}{.35cm}

\item Let $p_j=0$ for each item $j$, and set $j$ to be \textit{unmarked} 

\item Let $G=(U,V;E)$ and $G^+=(U^+,V^+;E)$ be the dynamic demand graph

\item While $G^+$ has a non-empty critical set
\begin{enumerate}
\item let $S\cup N(S)$ be a connected component in the subgraph induced by the critical set and its neighbors of $G^+$

\item for all items $j\in N(S)$, increase $p_j$ using \increaseprice\ on component $S\cup N(S)$; the new prices satisfy one of the following conditions:

\begin{enumerate}
\item[($\alpha$)] there are $i\in S$ and $j\notin N(S)$ (where either $b_{ij}>p_j$, or $b_{ij}=p_j$ and $j$ is unmarked) such that $i$ can get the same maximal utility from $j$ as it gets from items in $N(i)$
\item[($\beta$)] there are $i\in S$ and $j\in N(S)$ such that $v_{ij}=p_j$, i.e., $u_{ij}(p_j)=0$
\item[($\gamma$)] there are $i\in S$ and $j\in N(S)$ such that $b_{ij}=p_j$
\end{enumerate}

\item if the price of any marked item is strictly increased in the above step, \textit{unmark} the item
\item if condition ($\alpha$) is satisfied, make all such $(i,j)$ a new edge
\item if condition ($\beta$) is satisfied, delete all edges incident to all such $i$ in $G$

\item if neither condition ($\alpha$) nor ($\beta$) is satisfied but ($\gamma$) is satisfied
\begin{itemize}
\item set $j$ to be \textit{marked} for each $j\in N(S)$
\item for each pair $i\in S$ and $j\in N(S)$ satisfying $b_{ij}=p_{j}$, delete edge $(i,j)$ in $G$
\end{itemize}
\end{enumerate}

\item Define a reduced subgraph $G'$ from $G^+$ by deleting all edges $(i,j)\in E$ if $j$ is marked and all singleton vertices

\item If $G'$ has a non-empty critical set (denoted by $S'\subseteq U^+$)
\begin{enumerate}
\item set $j$ to be \textit{marked} for $j\in N'(S')$, where $N'(S')$ is the set of neighbors of $S'$ in $G'$
\item if there are $i\in S'$ and $j\in N'(i)$ such that $b_{ij}=p_j$
\begin{itemize}
\item for each pair $i\in S'$ and $j\in N'(i)$ satisfying $b_{ij}=p_j$, delete edge $(i,j)$ in $G$
\item goto Step~(3)
\end{itemize}
else goto Step~(6)
\end{enumerate}
\item Define a graph $G^*$ from $G^+$: for each $i\in U^+$ with at least one unmarked neighbor,\\ delete all edges connecting $i$ with a marked neighbor (i.e., delete $(i,j)\in E$ if $j$ is marked)


\item Extend $G^*$ to $H$ (with vertex set $U\cup V$) by adding all edges $(i,j)$ if $i\notin U^+$, $u_{ij}(p_j)=0$ \\ and $j$ is unmarked

\item If there is a maximum matching\footnote{} $M$ of $H$ (of size $m$) covering all buyers in $U^+$, output
\begin{enumerate}
\item an assignment of each buyer $i\in U$ according to $M$
\item a price for each $j\in V$ to be $p_j$ if $j$ is unmarked and $p_j+$ if $j$ is marked\footnote{}
\end{enumerate}

\item Else, return \textit{No Equilibrium Exists}\\[-0.15in]
\end{enumerate}

}}
\end{center}

\setlength{\baselineskip}{.5cm}

\footnotetext[7]{The existence of such maximum matching can be determined in polynomial time by, e.g., finding a maximum weighted matching of $H$ (assign a large weight for edges in $G^*$ and a small weight for edges in $H\setminus G^*$).}

\footnotetext[8]{When setting the price to be $p_j+$, we still need to keep the same preference for all marked items. The existence of such a vector is guaranteed by Proposition~\ref{prop-step-8b} and it can be computed by subroutine \increaseprice.}

To prove that the algorithm is correct, we need to prove two things:
\begin{itemize}
\item If the algorithm returns $(\mathbf{x},\mathbf{p})$, then it is a competitive equilibrium, and it is a minimum equilibrium.
\item If the algorithm does not return an output, there exists no competitive equilibrium.
\end{itemize}
Note that these statements imply immediately that if there exists a competitive equilibrium, there also exists a minimum equilibrium.

As discussed earlier, the algorithm only increases prices of items which are in the neighbor set of the critical set $S$ of $G^+(\mathbf{p})$ using the subroutine \increaseprice. 
This allows us to divide the analysis of the algorithm into stages, defined as follows. Divide the algorithm into \textit{phases} according to every execution of Steps~(3.e), (3.f) or (5,b), i.e., the deletion of any edge between the critical set and its neighbor set, because of either $v_{ij}=p_j$ (deleting all edges incident to $i$) or $b_{ij}=p_j$ (deleting edge $(i,j)$). Further, divide every phase into \textit{stages} according to every execution of Step~(3.d), i.e., the addition of edges between $S$ and $V\setminus N(S)$.
We have that the critical set of $G^+(\mathbf{p})$ remains the same within each stage until the final price at which that stage ends.

The following crucial lemma is proved for every stage in the algorithm, which guarantees that if the output $(\mathbf{x},\mathbf{p})$ returned by the algorithm is indeed an equilibrium, it is also a minimum equilibrium; this in particular implies that the price vector $\mathbf{p}$ at Steps (6, 7) of the algorithm is precisely the infimum price vector defined in Definition~\ref{def-min-eq}.

\begin{lemma}\label{lemma-minj}
Let $min_j$ be the minimum equilibrium price of item $j$ (if an equilibrium exists), and let $\mathbf{p}$ be the price vector at the end of any stage. Then, $p_j\le min_j$ for any $j$ for every stage in the algorithm. Further, if $j$ is marked, then $p_j<min_j$.
\end{lemma}

To show that the output of the algorithm is indeed a competitive equilibrium, we establish the following claims. Proposition~\ref{prop-no-critical} guarantees that every buyer with a nonempty demand set  (recall that the demand set contains only items leading to strictly positive utility) in graph $H$ defined in Step~(7) can be matched to an item in his demand set.

\begin{prop}\label{prop-no-critical}
The graph $G^*$ in Step (6) of the algorithm has no non-empty critical set. Thus, there exists a maximum matching in $H$ in which every buyer in $U^+$ is matched.
\end{prop}

The following fact, which simply follows from Step~(6) of the algorithm, guarantees that if a buyer in $U^+$ is matched to a marked neighbor in $H$, he will not prefer any unmarked items as long as the increment of prices of marked items in Step~(8.b) is sufficiently small.


\begin{fact}\label{fact-all-mark}
The items in the neighborhood $N^*(i)$ of a buyer $i \in U^+$ in graph $G^*$ are either all marked or all unmarked.
\end{fact}

In addition, we need to ensure that the prices of all marked items can be increased in Step~(8.b) in such a way that the demand structure does not change at the increased prices: without this, we cannot be sure that a matching in $H$ will indeed correspond to an equilibrium. (Of course, if we were to restrict ourselves to quasi-linear utility functions with budget constraints, then increasing all prices of marked items by the same small $\epsilon>0$ changes all utilities by the same amount, so that the structure of the demand graph is preserved; that the prices of marked items can be increased without changing the demand structure needs proof because we allow more general utility functions.)
The following lemma shows that there exists a strictly higher price vector inducing the same demand sets, so that no buyer strictly prefers one marked item in his demand set to another after the prices are increased. Thus, assigning a buyer in $U^+$ to {\em any} item in his demand set indeed maximizes his utility even after each item's price has been increased.
\begin{prop}\label{prop-step-8b}

Let $\mathbf{p}$ be the price vector when defining graph $H$ in Step~(7). Consider graph $H$: let $T$ be the set of marked items and $S$ be the neighbor set of $T$. Then, for any $\epsilon>0$, there exists $q_j = p_j+\epsilon_j$ for all $j \in T$, where $0<\epsilon_j<\epsilon$ can be arbitrarily small, and $q_j=p_j$ for $j\notin T$, such that $(i,j) \in G^+(\mathbf{q})$ if and only if $(i,j) \in G^+(\mathbf{p})$ for any $i\in S$ and $j\in T$.
\end{prop}

Putting these lemmas together ensures that the final allocation $\mathbf{x}$ matches all buyers in $U^+$ to an item which maximizes their utility. That is, the matching $\mathbf{x}$ returned by the algorithm satisfies all the envy-free conditions necessary for a competitive equilibrium.

Finally, requiring the matching defined by $\mathbf{x}$ to have size $m$ ensures that $(\mathbf{x},\mathbf{p})$ satisfies the market clearing condition as well. According to the definition of competitive equilibrium, we have the following conclusion.

\begin{lemma}\label{lemma-min-eq}
For any given instance of the problem, if \algmin\ outputs $(\mathbf{x},\mathbf{p})$, then it is a competitive equilibrium.
\end{lemma}

The following Lemma~\ref{lemma-min-exist} immediately implies the second statement needed for the proof of correctness, by proving the contrapositive. The proof of this lemma proceeds by showing, using Lemma~\ref{lemma-minj}, that the matching $\mathbf{x}$ defined by any competitive equilibrium $(\mathbf{x},\mathbf{p})$ must be contained in the graph $H$ in Step (8) of the algorithm; in which case the algorithm will return an output, yielding a contradiction.

\begin{lemma}\label{lemma-min-exist}
For any given instance of the problem, if a competitive equilibrium exists, then \algmin\ will output one.
\end{lemma}

The last result we need is about the runtime of the algorithm.

\begin{lemma}\label{lemma-polytime}
The algorithm \algmin\ runs in strongly polynomial time.
\end{lemma}

Therefore, Lemmas~\ref{lemma-min-eq} and~\ref{lemma-min-exist}, together with Lemma \ref{lemma-minj}, imply that \algmin\ produces a minimum competitive equilibrium if a competitive equilibrium exists for that instance, and this is done in strongly polynomial time by Lemma~\ref{lemma-polytime}. This gives us the main result Theorem~\ref{theorem-main}.


\section{Minimum Equilibrium Mechanism}\label{section-min-mechanism}

A natural question that arises in our assignment model is strategic behavior by buyers, since the utility function they report affects the final utility they receive from their allocation. Suppose we are given a marketplace, and a family of consistent utility functions such that any set of $mn$ utilities drawn from it are consistent (e.g., all quasi-linear utilities with budgets).
Consider the minimum equilibrium mechanism game, where
the strategy space of every buyer consists of all utility functions from this set, and
the private information of every buyer, as in~\cite{DG}, is the true utility functions over different items.
Given reported strategies/bids of utility functions from every buyer $i$ for every item $j$, the {\em minimum equilibrium mechanism} computes
a minimum competitive equilibrium if there exists one, and outputs nothing otherwise (\ie, all items remain unassigned).

The following claim, proved in Appendix~\ref{appendix-nash}, shows that truthful bidding is a Nash equilibrium if there does exist a competitive equilibrium with the true utility functions. The reason for not considering dominant strategy truthfulness is explained immediately after the theorem.

\begin{theorem}\label{theorem-nash}
Let $u_{ij}(\cdot)$ be the true (private) utility function of buyer $i$ for item $j$. If a competitive equilibrium exists when all buyers bid their true utility functions, then truthful bidding constitutes a Nash equilibrium in the minimum equilibrium mechanism.
\end{theorem}

The condition that a competitive equilibrium exists with truthful bids is necessary in the claim. This is because when an equilibrium does not exist, it is possible that buyers can submit untruthful modified utility functions for which a competitive equilibrium {\em does} exist, with an allocation that gives positive utility to all buyers (in contrast with $0$ utility when the mechanism does not allocate any items.) This implies that the minimum equilibrium mechanism is not dominant-strategy truthful, since whether truthful bidding is optimal or not depends on whether an equilibrium exists or not, which, of course depends on {\em all} the submitted utility functions: that is, the optimality of truthful bidding for a particular buyer cannot be independent of the bids submitted by other buyers.

\begin{exam}
Consider the quasi-linear utility with budgets model with three buyers $i_1,i_2,i_3$ and two items $j_1,j_2$. Values $v_{ij}$ and budgets $b_{ij} = b_i$ (all buyers have the same budgets for both items) are given below:
\begin{center}
\begin{tabular}{llll}
$v_{i_1,j_1} = 300$, & $v_{i_1,j_2} = 30$, & $b_{i_1} = 100$ \\
$v_{i_2,j_1} = 200$, & $v_{i_2,j_2} = 20$, & $b_{i_2} = 100$ \\
$v_{i_3,j_1} = 10$, & $v_{i_3,j_2} = 1$, & $b_{i_3} = \infty$ \\
\end{tabular}
\end{center}
In this example, when everyone bids truthfully, it is easy to verify that no equilibrium exists. However, if buyer $i_2$ bids, e.g., $v'_{i_2j_1} = 50$, $v'_{i_2j_2}= 20$ and $b'_{i_2}=100$, he will win $j_2$ in the minimum equilibrium mechanism at price $p_{j_2}=1$ (and $i_1$ wins $j_1$ at price $p_{j_1}=31$), from which he obtains a positive utility.
\end{exam}

\section{Conclusion}
In this paper, we presented a strongly polynomial time algorithm that decides whether or not a competitive equilibrium exists and if yes, computes a minimum one, for a general class of utility functions $u_{ij}$ with budgets in the assignment model. We note that the algorithm can be easily adapted to compute a maximal competitive equilibrium (i.e., all items are priced at the maximum among all equilibria) by a symmetric process of reducing prices.
It would be interesting to explore the algorithmic limits of our approach to compute competitive equilibria efficiently. The most natural and practically applicable, yet technically very challenging, extension is to the setting with multi-unit demand buyers and multi-unit supply sellers; we leave this as an open problem for future work.

\newpage

\appendix

\section{Examples}\label{appendix-example}

In this section, we list a number of examples used in the main context. Except for Example~\ref{example-no-min-eq}, the utilities are all quasi-linear $v_{ij}-p_j$ with budget constraints; also, unless otherwise specified, buyers have the same budget for all items.
The first example shows that in general a competitive equilibrium may not exist.

\begin{exam}\label{example-no-eq}
There are two buyers $i_1,i_2$ and one item $j$, with values $v_{i_1}=v_{i_2}=2$, and budgets $b_{i_1}=b_{i_2}=1$. When price $p_j\le 1$, both buyers desire the item; whereas when $p_j>1$, both of them vanish due to budget constraints.

A competitive equilibrium may not exist even if all values and budgets are different. For example, suppose there are three buyers $i_1,i_2,i_3$ and two items $j_1,j_2$, with values and budgets given below:
\begin{center}
\begin{tabular}{llll}
$v_{i_1,j_1} = 20$, & $v_{i_1,j_2} = 1$, & $b_{i_1} = 2$ \\
$v_{i_2,j_1} = 7$, & $v_{i_2,j_2} = 10$, & $b_{i_2} = \infty$ \\
$v_{i_3,j_1} = 0$, & $v_{i_3,j_2} = 30$, & $b_{i_3} = 5$
\end{tabular}
\end{center}
In this example, all values and budgets are different, yet there is no competitive equilibrium: there is no equilibrium when $p_{j_1}\le 2$ or $p_{j_2}\le 5$ because both items are over-demanded, but at any price beyond this at least one of the items remains unsold since both buyers $i_1$ and $i_3$ are priced out. Thus the condition that unsold items must be priced at zero cannot be satisfied.
\end{exam}

The following two examples show that even when a competitive equilibrium exists, a minimum equilibrium may not, for the two types of reasons discussed in Section~\ref{section-min}.
The first shows that if the utility functions are not strictly decreasing, even if continuity and consistency are satisfied, a minimum equilibrium does not exist because the equilibrium price vectors are incomparable. The second shows that the infimum of competitive equilibrium prices need not support a competitive equilibrium.

\begin{exam}\label{example-no-min-eq}
There are three buyers $i_1,i_2,i_3$ with infinite budget each and two items $j_1,j_2$. The utility functions are $u_{i_1j_1}(p_{j_1})=1-p_{j_1}$, $u_{i_2j_2}(p_{j_2})=1-p_{j_2}$, and $u_{i_3j_1}(p_{j_1})=u_{i_3j_2}(p_{j_2})=10$ (every undefined pair has negative utility for any price). There are two equilibrium price vectors $(1,0)$ (where $i_3$ wins $j_1$ and $i_2$ wins $j_2$) and $(0,1)$ (where $i_3$ wins $j_2$ and $i_1$ wins $j_1$), but there is no minimum equilibrium in this example.
\end{exam}

\begin{exam}\label{example-no-min-eq-iii}
There are two buyers $i_1$ and $i_2$ and one item $j$, with values $v_{i_1}=20$ and $v_{i_2}=100$, and budgets $b_{i_1}=3$ and $b_{i_2}=1$. Allocating the item to the first buyer at any price $p_j\in (1,3]$ is an equilibrium (buyer $i_2$ is envy-free due to his budget $b_{i_2}<p_j$). In this example, there is no exact minimum equilibrium, because there is no smallest real number bigger than 1. By our Definition~\ref{def-min-eq} of minimum equilibrium, allocating the item to buyer $i_1$ at price $p_j+$ where $p_j = 1$ is a minimum equilibrium.
\end{exam}

In the above Example~\ref{example-no-min-eq-iii}, allocating the item to buyer $i_1$ at price $p_j = 1$ is a buyer-optimal weakly stable matching~\cite{www09}:
buyer $i_2$, who is not envy-free in the competitive equilibrium concept, does not form any blocking pairs (in particular, $(i_2,j)$ is stable since $j$ cannot obtain more payment from $i_2$ due to his budget $b_{i_2}=p_j = 1$; this is illustrated by formula (4) in the definition of stability in~\cite{www09}). While this buyer-optimal weakly stable matching looks quite similar to the minimum equilibrium, the following example shows that they can be quite different in both allocations and prices.

\begin{exam}\label{example-weak-stable}
There are four buyers $i_1,i_2,i_3,i_4$ and three items $j_1,j_2,j_3$ with values and budgets given below (only $i_1$ has different budgets for different items):
\begin{center}
\begin{tabular}{llll}
$v_{i_1,j_1} = 100$, & $v_{i_1,j_2} = 50$, & $v_{i_1,j_3} = 0$, & $b_{i_1j_1} = 10$, \ $b_{i_1j_2} = b_{i_1j_3} =\infty$  \\
$v_{i_2,j_1} = 100$, & $v_{i_2,j_2} = 100$, & $v_{i_2,j_3} = 10$, & $b_{i_2} = 10$ \\
$v_{i_3,j_1} = 24$, & $v_{i_3,j_2} = 25$, & $v_{i_3,j_3} = 20$, & $b_{i_3} = \infty$ \\
$v_{i_4,j_1} = 0$, & $v_{i_4,j_2} = 0$, & $v_{i_4,j_3} = 100$, & $b_{i_4} = 5$
\end{tabular}
\end{center}
The buyer-optimal weakly stable matching computed by~\cite{www09} is $\{(i_1,j_1),(i_2,j_2),(i_3,j_3)\}$ at price vector $(10,10,5)$. The last buyer $i_4$, again due to his tight budget constraint $b_{i_4}=p_{j_3}=5$, does not form a blocking pair with item $j_3$. In the minimum equilibrium, however, the allocation is $\{(i_1,j_2),(i_2,j_3),(i_3,j_1)\}$ at price vector $(10+,11+,6+)$. Further, if $v_{i_3,j_1} = 0$ rather than 24 defined above, the buyer-optimal weakly stable matching remains the same but no competitive equilibrium exists.
\end{exam}

The following example shows that the solution returned by~\cite{lavi} is not a competitive equilibrium in our model (indeed, it can be much different from any competitive equilibrium in allocations and prices), and does not possess the efficiency property characteristic of competitive equilibria.
\begin{exam}\label{example-lavi}
Consider an example in~\cite{lavi} where there are two buyers $i_1,i_2$ and one item $j$ with values $v_{i_1}=7$ and $v_{i_2}=8$, and budgets $b_{i_1}=10$ and $b_{i_2}=7$. In the competitive equilibrium studied in this paper, $i_2$ wins the item at price $p_j=7$ (note that $i_1$ has utility 0 even if he wins the item). In the solution of~\cite{lavi} with strict boundary condition, $i_1$ wins the item at the same price. The second buyer $i_2$, who has $b_{i_2}=7$ dollars in pocket and deserves more value (i.e., $v_{i_2}=8$) for the item, is eliminated. As a result, this output is not a competitive equilibrium and its efficiency can be arbitrarily bad (note that the value $v_{i_2}$ can be arbitrarily large).


Further, the solution returned by~\cite{lavi} can be much different from any competitive equilibrium, even in the special sponsored search setting considered in~\cite{lavi}. Consider another example with three buyers $i_1,i_2,i_3$ and two items $j_1,j_2$ with values and budgets given below:
\begin{center}
\begin{tabular}{llll}
$v_{i_1,j_1} = 10$, & $v_{i_1,j_2} = 1$, & $b_{i_1} = \infty$ \\
$v_{i_2,j_1} = 100$, & $v_{i_2,j_2} = 10$, & $b_{i_2} = 10$ \\
$v_{i_3,j_1} = 50$, & $v_{i_3,j_2} = 5$, & $b_{i_3} = 5$
\end{tabular}
\end{center}
Note that the implicit click-through rates of $j_1$ and $j_2$ are 10 and 1, respectively.
The minimum competitive equilibrium in our paper has allocation $\{(i_2,j_1),(i_3,j_2)\}$ with price $(10,1)$. In the solution concept of~\cite{lavi} with strict budget constraint, the allocation is $\{(i_1,j_1),(i_2,j_2)\}$ with price $(10,5)$.
\end{exam}

The following example shows that the auction described in~\cite{LY} may not generate a competitive equilibrium, even if it exists.

\begin{exam}\label{example-LY}
There are four buyers $i_1,i_2,i_3,i_4$ and two items $j_1,j_2$ with values and budgets given below:
\begin{center}
\begin{tabular}{llll}
$v_{i_1,j_1} = 10$, & $v_{i_1,j_2} = 0$, & $b_{i_1} = 1$ \\
$v_{i_2,j_1} = 10$, & $v_{i_2,j_2} = 0$, & $b_{i_2} = 1$ \\
$v_{i_3,j_1} = 5$, & $v_{i_3,j_2} = 10$, & $b_{i_3} = \infty$ \\
$v_{i_4,j_1} = 5$, & $v_{i_4,j_2} = 10$, & $b_{i_4} = \infty$
\end{tabular}
\end{center}
Assume that reserve prices are $p_{j_1}=p_{j_2}=0$. In this example, a minimum competitive equilibrium is $x_{i_3}=j_1$ and $x_{i_4}=j_2$ at price vector $(1+,6+)$.
In the algorithm of~\cite{LY}, we increase $p_{j_1}$ (which is a minimal over-demanded set) to 2 and it becomes under-demanded (since $i_1, i_2$ have no enough budget and $i_3, i_4$ prefer $j_2$ at this moment); thus we set $p_{j_1}=1$ and allocate it to one of $i_1$ and $i_2$ permanently. Next since $j_2$ is over-demanded as well, we increase its price and eventually it is allocated to one of $i_3$ and $i_4$ at price $p_{j_2}=10$. This outcome is only a rationed equilibrium~\cite{LY} but not a competitive equilibrium as both $i_3$ and $i_4$ strictly prefer item $j_1$.
\end{exam}


\section{Critical Set: Proof of Theorem~\ref{them-critical}}\label{appendix-theorem-critical}

\begin{proof}
We prove the two claims respectively.

\begin{itemize}
\item We first prove the first part. Let $G'$ denote the resulting graph after adding edges between $A$ and $N(A)$. For any subset of vertices $B\subseteq U$, let $N'(B)$ denote the neighbor set of $B$ in $G'$. Note that we have $N(B)\subseteq N'(B)$; thus $|B|-|N'(B)|\le |B|-|N(B)|$. Further, for the critical set $A$ of $G$, by the rule of adding edges, we have $N(A)=N'(A)$; thus $|A|-|N(A)| = |A|-|N'(A)|$. Therefore, $G$ and $G'$ have the same deficiency, i.e., $\delta(G)=\delta(G')$. For any subset $A'\subseteq A$, we have
    \[|A|-|N'(A)| = |A|-|N(A)| > |A'|-|N(A')| \ge |A'|-|N'(A')|\]
    where the strict inequality follows from the fact that $A$ is the critical set of $G$. Hence, $A$ is the critical set of $G'$ as well.

\item For the second part, assume otherwise that $A'\subseteq U$, $A'\neq A$, is the critical set of graph $G'$, where $G'$ is obtained from $G$ by deleting some edges between $U\setminus A$ and $N(A)$. In the following discussions, for any subset $B\subseteq U$, $N(B)$ and $N'(B)$ denote the set of neighbors of $B$ in $G$ and $G'$, respectively. Note that $N'(B)\subseteq N(B)$.

    Let $X_1 = A'\cap A$ and $X_2=A'\setminus X_1$ be a partition of $A'$. Assume that $X_2\neq \emptyset$. Let
    \[ Y = \{j\in V\setminus N(A)~|~\exists \ i\in X_2 \ s.t. \ (i,j)\in E\}\]
    be the set of neighbors of $X_2$ which are not in $N(A)$ (in both $G$ and $G'$).
    Since $X_1\subseteq A$, we have $N(X_1)=N'(X_1)\subseteq N(A)$. For $X_2$, it can be seen that $|X_2|\le |Y|$. This is because, otherwise
    \[|A\cup X_2| - |N(A\cup X_2)| = |A|+|X_2| - (|N(A)| + |Y|) > |A|-|N(A)|\]
    which contradicts to the fact that $A$ has maximal deficiency in $G$. Hence,
    \[|X_1| - |N'(X_1)| \ge |X_1| - |N'(X_1)| + |X_2| - |Y| \ge |X_1\cup X_2| - |N'(X_1\cup X_2)| = |A'|-|N'(A')|\]
    which contradicts to the assumption that $A'$ is the critical set of $G'$.

    Hence, $X_2=\emptyset$ and $A'\subseteq A$, which implies that
    \[|A'|-|N(A')| = |A'|-|N'(A')| \ge |A|-|N'(A)| = |A|-|N(A)|\]
    where the inequality follows from the fact that $A'$ is the critical set of $G'$. Since $A$ is the critical set of $G$, we must have $A'=A$.
\end{itemize}
\end{proof}

\section{Analysis of \increaseprice}\label{appendix-alg-increase}

Theorem~\ref{theorem-price-increase} follows immediately from the following three lemmas. In all lemmas and their proofs, $\mathbf{p}$ denotes the initial price vector at Step~(1) of \increaseprice.

\begin{lemma}\label{lemma-price-increase1}
For any edge $(i_0,j_0) \in C$ and item $j\in N(S)$, the value $q^{i_0j_0}_{j}$ is well-defined in $\mathbf{q}^{i_0j_0}$ in Step~(3.c) and $\mathbf{q}^{i_0j_0}\ge \mathbf{p}$. Further, in the process of the Step~(3) where $q^{i_0j_0}_{j}$ has not been defined yet, for any $(i,j),(i,j')\in C$ with $q^{i_0j_0}_{j'}$ already being defined and $(i',j),(i',j'')\in C$ with $q^{i_0j_0}_{j''}$ already being defined (i.e., there are different choices $i$ or $i'$ to define $q^{i_0j_0}_{j}$), the defined value $q^{i_0j_0}_{j}$ will be the same for either choice.
\end{lemma}
\begin{proof}
Let $T=N(S)$. Note that $q^{i_0j_0}_{j_0}=q^{i_0j_0}$ is well-defined by the subroutine. Let $T'\subseteq T$ be the subset items whose values $q^{i_0j_0}_{j}$ are not defined. Then for any buyer $i\in S$, its set of neighbors $N(i)$ in $N(S)$ is either in $T'$ or $T\setminus T'$. Let $S'=\{i\in S~|~N(i)\subseteq T'\}$. Then $S'\cup T'$ and $(S\setminus S')\cup (T\setminus T')$ are disconnected, a contradiction to the fact that $C$ is a connected component.

Given the initial price vector $\mathbf{p}=(p_1,\ldots,p_m)$, each buyer obtains his maximal positive utility from its neighbor items in $T$; thus $u_{i_0j_0}(p_{j_0})>u^{i_0}_{\max}$. Since the utility function $u_{ij}(\cdot)$ is decreasing, by the definition of $v_{ij}$ and $b_{ij}$, we have $p_{j_0}\le \min\{v_{i_0j_0}, b_{i_0j_0}, u^{-1}_{i_0j_0}(u^{i_0}_{\max})\}=q^{i_0j_0}=q^{i_0j_0}_{j_0}$. Then by induction on the process of defining values $q^{i_0j_0}_{j}$  in the algorithm, we have $\mathbf{q}^{i_0j_0}\ge \mathbf{p}$. (Note that if any $q^{i_0j_0}_{j}=\infty$, then certainly $\mathbf{q}^{i_0j_0} = \infty \ge \mathbf{p}$.)

It remains to show that the value $q^{i_0j_0}_{j}$ is uniquely defined. We use induction on the order of item $j\in N(S)$ in which $q^{i_0j_0}_{j}$ is defined. By the rule of Step~(3) of \increaseprice, the value $q^{i_0j_0}_{j_0}$ is uniquely defined. For each item $j$ whose value $q^{i_0j_0}_{j}$ is about to define in the process, if there are two possible ways (the argument for more than two ways is the same) to define $q^{i_0j_0}_{j_0}$, say $(i,j),(i,j')\in C$ with $q^{i_0j_0}_{j'}$ being defined and $(i',j),(i',j'')\in C$ with $q^{i_0j_0}_{j''}$ being defined, assume without loss of generality that $q^{i_0j_0}_{j_0}$ is defined in terms of buyer $i$, i.e., $q^{i_0j_0}_{j}=u_{ij}^{-1}(u_{ij'}(q^{i_0j_0}_{j'}))$. Since all utility functions are strict decreasing, it suffices to show that $u_{i'j}(q^{i_0j_0}_{j})=u_{i'j''}(q^{i_0j_0}_{j''})$. Consider the path starting from $j, i, j'$ to $j_0$ according to the backward order of defined values, and then from $j_0$ to $j''$ according to the forward order of defined values. By induction and the definition of $q^{i_0j_0}_{j}$, this is a transitive path in both $\mathbf{p}$ and $\mathbf{q}^{i_0j_0}$. Thus by the consistency property, $i'$ has the same preference over $j$ and $j''$, i.e., $u_{i'j}(q^{i_0j_0}_{j})=u_{i'j''}(q^{i_0j_0}_{j''})$. This completes the proof of the claim.
\end{proof}

\begin{lemma}\label{lemma-price-increase2}
For any edge $(i_0,j_0) \in C$, if $\mathbf{q}^{i_0j_0}\neq \infty$, then for any buyer $i\in S$ and $j\in N(i)$, where $N(i)$ is the neighborhood of $i$ in $G^+(\mathbf{p})$, we have $u_{ij}(q^{i_0j_0}_{j}) \ge u_{ij'}(q^{i_0j_0}_{j'})$ for any other item $j'\in N(S)$, i.e., $i$ still weakly prefers $j$ to all other items in $N(S)$ at price $\mathbf{q}^{i_0j_0}$. In particular, this implies that for any $j,j'\in N(i)$, $u_{ij}(q^{i_0j_0}_{j})=u_{ij'}(q^{i_0j_0}_{j'})$.
\end{lemma}
\begin{proof}
By the subroutine \increaseprice\ of defining new price vectors, for any $\mathbf{q}^{i_0j_0}\neq \infty$, we have $q^{i_0j_0}_{j}\le q^{ij} \le \min\{v_{ij},b_{ij}\}$ for any $(i,j)\in C$. Thus all prices $q^{i_0j_0}_{j}$ defined by the algorithm are upper bounded by $v_{ij},b_{ij}$ for any edge $(i,j)\in C$.

Consider any buyer $i$ and two items $j,j'$ where $j\in N(i)$. Let $j_0,i_1,j_1,\ldots,i_\ell,j_\ell=j$ be the path to define $q^{i_0j_0}_{j}$ in Step~(3.c), where $(i_k,j_{k-1}),(i_k,j_k)\in C$ for $k=1,\ldots,\ell$. That is, $q^{i_0j_0}_{j_k}$ is determined according to $q^{i_0j_0}_{j_{k-1}}$ in Step~(3.c). Hence, we have $u_{i_kj_{k-1}}(q^{i_0j_0}_{j_{k-1}}) = u_{i_kj_{k}}(q^{i_0j_0}_{j_k})$. Similarly, if $j_0,i_{1'},j_{1'},\ldots,i_{\ell'},j_{\ell'}=j'$ is the path to define $q^{i_0j_0}_{j'}$ in Step~(3.c), where $(i_{k'},j_{(k-1)'}),(i_{k'},j_{k'})\in C$ for $k'=1',\ldots,\ell'$, then we have $u_{i_{k'}j_{(k-1)'}}(q^{i_0j_0}_{j_{(k-1)'}}) = u_{i_{k'}j_{{k'}}}(q^{i_0j_0}_{j_{k'}})$. Now consider putting the two paths together: $j=j_\ell, i_\ell, \ldots,j_1,i_1,j_0, i_{1'},j_{1'},\ldots,i_{\ell'},j_{\ell'}=j'$, which is transitive in both $\mathbf{p}$ and $\mathbf{q}^{i_0j_0}$. Since $j\in N(i)$, we have $u_{ij}(p_j)\ge u_{ij'}(p_{j'})$. By the consistency property, we have $u_{ij}(q^{i_0j_0}_{j}) \ge u_{ij'}(q^{i_0j_0}_{j'})$.
\end{proof}

The above lemma, as well as the following Lemma~\ref{lemma-price-increase3}, implies that in the minimum price vector $\mathbf{q}$ (defined in terms of one of $\mathbf{q}^{i_0j_0}\neq \infty$), all buyers in $S$ still weakly prefer their old neighbors in $C$ to all other items with respect to $\mathbf{q}$ (including those items not in $N(S)$, since $q_{j_0}\le q^{i_0j_0}$ as Lemma~\ref{lemma-price-increase3} proves, no buyer will strictly prefer an item not in $N(S)$), yielding the third claim of Theorem~\ref{theorem-price-increase}. The following lemma gives the proof of the second claim of Theorem~\ref{theorem-price-increase}.

\begin{lemma}\label{lemma-price-increase3}
For any $(i_0,j_0),(i'_0j'_0)\in C$, $\mathbf{q}^{i_0j_0}$ and $\mathbf{q}^{i'_0j'_0}$ are comparable. Further, the minimum vector $\mathbf{q}$ is well-defined and satisfies $q_{j_0}\le q^{i_0j_0}$ for all $(i_0,j_0)\in C$, where $q^{i_0j_0}$ is defined in Step~(3.b) of the algorithm. In particular, this implies that $\mathbf{q}\neq \infty$.
\end{lemma}
\begin{proof}
Similar to the above proof, we assume that all prices defined by the algorithm are upper bounded by $v_{ij},b_{ij}$ for any edge $(i,j)\in C$. Consider any $\mathbf{q}^{i_0j_0}, \mathbf{q}^{i'_0j'_0}\neq \infty$, where $(i_0,j_0),(i'_0j'_0)\in C$, we first show that $\mathbf{q}^{i_0j_0}$ and $\mathbf{q}^{i'_0j'_0}$ are comparable.
Assume without loss of generality that $q^{i_0j_0}_{j_0}\le q^{i'_0j'_0}_{j_0}$ (the same analysis below holds when $q^{i_0j_0}_{j_0}\ge q^{i'_0j'_0}_{j_0}$). For any item $j\in N(S)$, consider defining vector $\mathbf{q}^{i_0j_0}$ in the algorithm and let $j_0,i_1,j_1,\ldots,i_\ell,j_\ell=j$ be the path to define $q^{i_0j_0}_{j}$, where $(i_k,j_{k-1}),(i_k,j_k)\in C$ for $k=1,\ldots,\ell$. That is, $q^{i_0j_0}_{j_k}$ is determined according to $q^{i_0j_0}_{j_{k-1}}$ in Step~(3.c). Hence, we have $u_{i_kj_{k-1}}(q^{i_0j_0}_{j_{k-1}}) = u_{i_kj_{k}}(q^{i_0j_0}_{j_k})$.  If $q^{i_0j_0}_{j_{1}} > q^{i'_0j'_0}_{j_{1}}$, then by the strict monotonicity of utility functions, we have
\[u_{i_1 j_{1}}(q^{i'_0j'_0}_{j_{1}})> u_{i_1 j_{1}}(q^{i_0j_0}_{j_{1}}) = u_{i_1 j_{0}}(q^{i_0j_0}_{j_{0}})\ge u_{i_1 j_{0}}(q^{i'_0j'_0}_{j_{0}})\]
This contradicts to Lemma~\ref{lemma-price-increase2} which implies that $u_{i_1 j_{0}}(q^{i'_0j'_0}_{j_{0}})=u_{i_1 j_{1}}(q^{i'_0j'_0}_{j_{1}})$ given $(i_1,j_{0}), (i_1,j_{1})\in C$. Therefore, $q^{i_0j_0}_{j_{1}} \le q^{i'_0j'_0}_{j_{1}}$. We can do the same analysis through all edges along the path, and at the end conclude that $q^{i_0j_0}_{j} = q^{i_0j_0}_{j_{\ell}} \le q^{i'_0j'_0}_{j_{\ell}} = q^{i'_0j'_0}_{j}$. Therefore, for any item $j\in N(S)$, $q^{i_0j_0}_{j}\le q^{i'_0j'_0}_{j}$, i.e., $\mathbf{q}^{i_0j_0}\le \mathbf{q}^{i'_0j'_0}$.

It remains to show that there is $(i_0,j_0)\in C$ such that $\mathbf{q}^{i_0j_0}\neq \infty$. For any item $j$, we say its value $q^{i_0j_0}_{j}$ {\em feasible} if for any $i\in N(j)$, $q^{i_0j_0}_{j} \le q^{ij} = \min\{v_{ij},b_{ij},u^{-1}_{ij}(u^{i}_{\max})\}$. For simplicity, we assume without loss of generality that in Step~(3.c) of the algorithm, it always tries to define as many feasible values as possible before reaching any item $j$ having an infeasible value (from which the algorithm sets $\mathbf{q}^{i_0j_0}= \infty$). This will only affect the execution of Step~(3.c), but not the outcome of the algorithm.

Consider any $(i_0,j_0)\in C$, if $\mathbf{q}^{i_0j_0} = \infty$, let $j_1$ be the first item with an infeasible value set by the algorithm and $T_1$ be the set of items whose (feasible) values $q^{i_0j_0}_{j}$ have been defined at that moment. Note that $j_1\notin T_1$. Next let $i_1=\arg\min_{i\in N(j_1)}q^{ij_1}$ and consider defining vector
$\mathbf{q}^{i_1j_1}$ in the algorithm. Similarly, if $\mathbf{q}^{i_1j_1} = \infty$, let $j_2$ be the first item with an infeasible value set by the algorithm and $T_2$ be the set of items whose (feasible) values $q^{i_1j_1}_{j}$ have been defined at that moment. Note that $j_1\in T_2$ and $j_2\notin T_2$. Similar to the argument in the first part of the proof, since $q^{i_1j_1}_{j_1} = q^{i_1j_1} < q^{i_0j_0}_{j_1}$, we can show that for any $j\in T_1$, $q^{i_1j_1}_{j}\le q^{i_0j_0}_{j}$; thus $T_1\subset T_2$. Next let $i_2=\arg\min_{i\in N(j_2)}q^{ij_2}$ and consider defining vector $\mathbf{q}^{i_2j_2}$ in the algorithm and the process continues: for any $T_k$ associated with vector $\mathbf{q}^{i_k,j_k}$, value $q^{i_kj_k}_{j}$ is feasible for any $j\in T_k$. In the process, we always increase the number of items with feasible values. Therefore eventually we will reach to a vector such that all items have feasible values. Hence, in the minimum price vector $\mathbf{q}$, all prices are feasible, i.e., $q_j\le q^{ij}$ for any $(i,j)\in C$.
\end{proof}

The following claim gives an equivalent way to view the subroutine \increaseprice: instead of increasing all prices from $\mathbf{p}$ to $\mathbf{q}$ directly, it can be decomposed into a ``continuous" process where all prices are increased continuously and simultaneously. That is, increase from $\mathbf{p}$ to $\mathbf{r}$, and then from $\mathbf{r}$ to $\mathbf{q}$
(the process can be divided further into arbitrarily small amount of increment). This idea will be helpful for our analysis in the following sections.

\begin{prop}\label{prop-price-increase}
Assume that subroutine \increaseprice\ increases prices of items $N(S)$ from $\mathbf{p}$ to $\mathbf{q}$, where $\mathbf{p}<\mathbf{q}$. For each item $j\in N(S)$, let $min_{j}$ be any value satisfying $p_{j}< min_{j}<q_{j}$. If we add an extra condition in Step~(3.b) of the subroutine by requiring that $q^{ij}\le min_{j}$ for every $(i,j)\in C$, then the subroutine will outputs a minimum price vector (denoted by $\mathbf{r}$) that satisfies $\mathbf{p}<\mathbf{r}<\mathbf{q}$, and all properties regarding $\mathbf{q}$ in Theorem~\ref{theorem-price-increase} still hold for $\mathbf{r}$. Further, there is $j_0\in N(S)$ such that $r_{j_0}=min_{j_0}$.
\end{prop}
\begin{proof}
The introduction of $min_j$ for each item $j$ gives another upper bound on $q^{ij}$ for any $(i,j)\in C$. The subroutine and analysis of the above Lemma~\ref{lemma-price-increase1},~\ref{lemma-price-increase2} and~\ref{lemma-price-increase3} remain the same (except whenever we talk about the upper bound of $q^{ij}$, say the definition of feasibility in the proof of Lemma~\ref{lemma-price-increase3}, this new bound $min_j$ should be included as well), and thus all properties regarding $\mathbf{q}$ in Theorem~\ref{theorem-price-increase} still hold for $\mathbf{r}$.

For the new implementation with the extra condition $q^{ij}\le min_{j}$, we use symbol `r' instead of `q' to denote all computed prices.
Consider any edge $(i,j)\in C$ and the computation of $\mathbf{r}^{ij}$, i.e., the implementation of the algorithm on edge $(i,j)$. Since $min_{j}<q_{j}\le q^{ij} = \min\{v_{ij}, b_{ij}, u^{-1}_{ij}(u^{i}_{\max})\}$, where the second inequality follows from Lemma~\ref{lemma-price-increase3}, in Step~(3.b) of the new implementation, we have $r^{ij}=min_{j}$. Therefore,
if $\mathbf{r}^{ij}\neq \infty$, we have $r^{ij}_{j}=r^{ij}=min_{j}$. Since $\mathbf{r}\neq \infty$ by Lemma~\ref{lemma-price-increase3}, the minimum price vector $\mathbf{r}$ is obtained at some $\mathbf{r}^{i_0j_0}\neq \infty$ where $(i_0,j_0)\in C$; thus $r_{j_0}=r^{i_0j_0}_{j_0}=min_{j_0}$. This implies that $\mathbf{r}<\mathbf{q}$.
Since $\mathbf{p}<\mathbf{q}$ and $p_{j}< min_{j}<q_{j}$ for any $j\in N(S)$, we have $\mathbf{p}<\mathbf{r}$ by the consistency property. Therefore, $\mathbf{p}<\mathbf{r}<\mathbf{q}$.
\end{proof}

\section{Analysis of \algmin}\label{section-min-analysis}

In this section, we will prove all claims in Section~\ref{section-alg-main}, which gives the proof of the main Theorem~\ref{theorem-main}. We will first prove Proposition~\ref{prop-no-critical} and~\ref{prop-step-8b}.


\medskip \noindent \textbf{Proposition~\ref{prop-no-critical}.} The graph $G^*$ in Step (6) of the algorithm has no non-empty critical set. Thus, there exists a maximum matching in $H$ in which every buyer in $U^+$ is matched.

\begin{proof}
Consider the graph $G^*$ in Step~(6). For any $i\in U$ with $N^*(i)\neq \emptyset$ (recall that $N^*(S)$ is the neighbor set of $S$ in graph $G^*$), all items in $N^*(i)$ are either marked or unmarked.
Define
\[X_1=\{i\in U~|~N^*(i)\neq \emptyset \ \textup{and all items in $N^*(i)$ are marked}\}\]
and
\[X_2=\{i\in U~|~N^*(i)\neq \emptyset \ \textup{and all items in $N^*(i)$ are unmarked}\}\]
Assume that $S$ is the critical set of $G^*$. Let $S_1=S\cap X_1$ and $S_2=S\cap X_2$. Note that $S=S_1\cup S_2$ and $N^*(S_1)\cap N^*(S_2)=\emptyset$.

We claim that $|S_2| \le |N^*(S_2)|$. Consider the last execution of the algorithm in the Step~(5) right before moving to Step~(6). There are two possibilities that the algorithm moves to Step~(6), the ``if"-condition in main statement of Step~(5) fails or the ``if"-condition in Step~(5.b) fails. We consider each of them separately.
\begin{itemize}
\item The ``if"-condition fails in the main statement of Step~(5). It implies that $G'$ has no non-empty critical set, and hence $|S_2|\le |N'(S_2)|$, where $N'(S_2)$ is the neighbor set of $S_2$ in graph $G'$. At the same time, Step~(5) will be skipped. Hence, all edges in $G'$ are in $G^*$, which implies that $|S_2| \le |N^*(S_2)|$.

\item The ``if"-condition holds in the main statement of Step~(5) and the ``if"-condition fails in Step~(5.b), i.e., the algorithm transits to Step~(6) after executing Step~(5.a). Let $S'$ be the critical set of $G'$ in that execution. Since all items in $N'(S')$ only get marked here, all edges incident to $S'$ are kept and all items in $N^*(S')$ are marked in $G^*$. Hence, $S_2\subseteq U^+\setminus S'$. Observe that $N^*(S_2)=N'(S_2) \setminus N'(S')$ (because for any $j\in N^*(S_2)$, $j$ is unmarked, and thus, $j\in N'(S_2)\setminus N'(S')$; on the other hand, for any $j\in N'(S_2) \setminus N'(S')$, $j$ is not marked by Step~(5), hence $j\in N^*(S_2)$). Therefore,
    \[|S'\cup S_2| - |N'(S'\cup S_2)| = |S'|+|S_2|-|N'(S')|-|N^*(S_2)| > |S'|-|N'(S')|\]
    where the last inequality is by assumption. This contradicts to the fact that $S'$ is the critical set of $G'$.
\end{itemize}

Hence, $|S_2| \le |N^*(S_2)|$, which implies that \[|S_1|-|N^*(S_1)| \ge |S_1|+|S_2|-|N^*(S_1)|-|N^*(S_2)| = |S|-|N^*(S)|\]
Thus, we must have $S_2=\emptyset$ by the minimality of the critical set $S$, i.e., all items in the neighborhood of $S$ in $G^*$ are marked.
By the rule of defining $G^*$ in Step~(6), all items in the neighborhood of $S$ in $G^+$ are marked as well.

Therefore, for any neighbor $j$ of $S$ in the last run of Step~(3), if $j$ is marked then $j$ is a neighbor of $S$ in $G^*$; if $j$ is unmarked then it must be converted to marked in Step~(5.a), otherwise some buyer in $S$ would have unmarked neighbor in $G^+$ and thus would only have unmarked neighbor in $G^*$, which contradicts to the fact that $S\subseteq X_1$. Therefore, all neighbors of $S$ in the last run of Step~(3) remain to be the neighbors of $S$ in $G^*$. Since there is no critical set when the algorithm gets out of Step~(3), we know that $|S|$ is less than or equal to the size of its neighborhood in $G^+$, thus $|S|\le |N^*(S)|$. This implies $S$ cannot be the critical set of $G^*$.
\end{proof}


\medskip \noindent \textbf{Proposition~\ref{prop-step-8b}.}
Let $\mathbf{p}$ be the price vector when defining graph $H$ in Step~(7). Consider graph $H$: let $T$ be the set of marked items and $S$ be the neighbor set of $T$. Then, for any $\epsilon>0$, there exists $q_j = p_j+\epsilon_j$ for all $j \in T$, where $0<\epsilon_j<\epsilon$ can be arbitrarily small, and $q_j=p_j$ for $j\notin T$, such that $(i,j) \in G^+(\mathbf{q})$ if and only if $(i,j) \in G^+(\mathbf{p})$ for any $i\in S$ and $j\in T$.

\begin{proof}
By the construction of graph $G^*$ and $H$ in Step~(6, 7) of the algorithm, for any buyer $i \in S$ its neighborhood $N^*(i)$ is either all marked or all unmarked. Hence, $N(S)=T$. Note that all new added edges in Step~(7) only connect to unmarked items, i.e., not incident to $T$.
Consider any $i\in S$ and items $j,j'$, where $j\in N(i)$ and $j'\notin N(i)$, we have $u_{ij}(p_j) > u_{ij'}(p_{j'})$. Since all utility functions are continuously decreasing, there is $\delta_j$, $0<\delta_j<\epsilon$, such that $u_{ij}(p_j+\epsilon_j) > u_{ij'}(p_{j'})$ for all $0<\epsilon_j\le \delta_j$ (note that the inequality is guaranteed due to $p_j<\min\{v_{ij},b_{ij}\}$ since $j\in T$ is marked). Actually, we can pick $\delta_j$ to be sufficiently small so that the above inequality holds for any $i\in N(j)$ and $j'\notin N(i)$. Now we can apply Proposition~\ref{prop-price-increase} to all items in $T$ (i.e., set up an upper bound $\delta_j$ to increase), from which we can increase all $p_j$ for $j\in T$ by an arbitrarily small amount and the demand graph structure remains the same.
\end{proof}

Next we will prove Lemma~\ref{lemma-minj},~\ref{lemma-min-eq},~\ref{lemma-min-exist} and~\ref{lemma-polytime}.


\medskip \noindent \textbf{Lemma~\ref{lemma-minj}.} Let $min_j$ be the minimum equilibrium price of item $j$ (if an equilibrium exists), and let $\mathbf{p}$ be the price vector at the end of any stage. Then, $p_j\le min_j$ for any $j$ for every stage in the algorithm. Further, if $j$ is marked, then $p_j<min_j$.

\begin{proof}
We will prove the claim by induction on stages. At the beginning all prices are 0 and the claim follows obviously. Consider any stage in the algorithm. Assume that the claim holds at the beginning of the stage, and we will show that the claim still holds at the end of the stage. Let $C=S\cup T$ be the connected component of the critical set specified by Step~(3.a) in that stage, i.e., implemented by subroutine \increaseprice, where $S$ is the set of buyers and $T=N_C(S)$ is its neighbor set, where $N_C(S)$ is the neighbor set of $S$ in $C$. Denote by $\mathbf{p}=(p_j)_{j\in T}$ the price vector of items in $T$ at the end of the stage by \algmin.

We first prove that $p_j\le min_j$ for any item $j\in T$. Assume otherwise that there is $j\in T$ such that $p_j>min_j$. Consider price vector $\mathbf{q}=(q_j)_{j\in T}$ returned by subroutine \increaseprice\, with the initial setup equal to the price vector at the beginning of the stage, plus an extra condition $q^{i_0j_0}\le min_{j_0}$ in Step~(3.b) of the subroutine for all edges $(i_0,j_0)$. By applying Proposition~\ref{prop-price-increase}, we know that such minimum price vector $\mathbf{q}$ exists, and its prices are guaranteed to be upper bounded by the minimum equilibrium price vector, i.e., $q_j\le min_j$ for any $j\in T$. Further, there is $j_0\in T$ such that $q_{j_0}=min_{j_0}$, and $\mathbf{q}<\mathbf{p}$ (note that this $\mathbf{p}$ is not the initial price vector when we run \increaseprice, but the final price vector of the considered stage). Let $D$ denote the demand graph of buyers $S$ at price vector $\mathbf{q}$, and $N_D(i)$ denote the neighbor set of any $i\in S$ in $D$. Note that $N_D(i)\subseteq T$, in particular, $N_D(S)\subseteq T$; and by Lemma~\ref{lemma-price-increase2}, $C\subseteq D$, i.e., $C$ is a subgraph of $D$.

Let $(\mathbf{x}^*,\mathbf{p}^*)$ be an equilibrium with $p^*_{j_0}=min_{j_0}$ and $M^*$ be the corresponding matching. By the above argument, we have $q_j\le min_j\le p^*_j$ for any $j\in T$, and $q_{j_0}=min_{j_0}=p^*_{j_0}$. Hence,
\[u_{ij}(q_j) \ge u_{ij}(p^*_j), \ \forall\ i\in S, j\in T\]
Let $Z\subseteq S$ be a subset of buyers to which there is an alternating path from $j_0$ in $D$ in terms of matching $M^*$ (i.e., edges are not in $M^*$ and in $M^*$ alternatively).
Note that $Z\neq \emptyset$ since for any item $j\in T$, it must have at least two neighbors in $S$ in $C$ (otherwise, we can find a smaller critical set).

We claim that for each $i\in Z$, $x^*_i\in N_D(i)$ (i.e., $i$ wins an item in $N_D(i)$ in the equilibrium $(\mathbf{x}^*,\mathbf{p}^*)$). This can be shown by induction on the distance between $i$ and $j_0$ in the alternating path. If the distance is one, i.e., $j_0\in N_D(i)$, since $q_{j_0}<p_{j_0}$, $i$ can get a positive utility from $j_0$ at price vector $\mathbf{q}$. Since $q_{j_0}=p^*_{j_0}$, $i$ must be a winner in the equilibrium and its utility satisfies
\[u_{ix^*_i}(p^*_{x^*_i}) \ge u_{ij_0}(p^*_{j_0}) = u_{ij_0}(q_{j_0})\]
Hence,
\[u_{ix^*_i}(q_{x^*_i})\ge u_{ix^*_i}(p^*_{x^*_i})\ge u_{ij_0}(q_{j_0})\]
That is, the utility that $i$ can get from $x^*_i$ is at least as large as its maximal utility in $D$.
Hence, $x^*_i\in N_D(i)$ and all inequalities in the above are tight. In particular, this implies that $p^*_{x^*_i}=q_{x^*_i}$ as the utility function is strictly decreasing. If the distance between $i$ and $j_0$ is 3 via an alternating path, say $(i,j_1,i_0,j_0)$ where $i_0$ wins item $j_1$ in the equilibrium, then what we just showed implies that $p^*_{j_1}=q_{j_1}$. Since $j_1\in N_D(i)$, the situation of $j_1$ for $i$ is the same that of $j_0$ for $i_0$. Hence, we can apply the same argument to conclude that $x^*_i\in N_D(i)$ and $p^*_{x^*_i}=q_{x^*_i}$. Therefore, we can show inductively that the claim holds for all buyers in $Z$ with distance $3,5,7,\ldots,$ to $j_0$.

Hence, in the equilibrium $(\mathbf{x}^*,\mathbf{p}^*)$, each $i\in Z$ wins item $x^*_i\in N_D(i)$. Therefore, $\{(i,x^*_i)~|~i\in Z\}$ defines a matching. Let $Y=\{x^*_i~|~i\in Z\}$. Since $|S|>|T|$ and $Y\subseteq N_D(S)\subseteq T$, clearly $S\setminus Z\neq \emptyset$. Note that for any $i\in S\setminus Z$, there is no edge between $i$ and $Y$ in $D$ (otherwise, there would be an alternating path between $j_0$ and $i$, and we would put $i$ into $Z$). That is, $N_D(S\setminus Z)\subseteq T\setminus Y$. Hence,
\begin{eqnarray*}
|S\setminus Z| - |N_C(S\setminus Z)| &\ge& |S\setminus Z| - |N_D(S\setminus Z)| \\
&\ge& |S\setminus Z| - |T\setminus Y| \\
&=& |S\setminus Z| + |Z| - |T\setminus Y| - |Y| \\
&=& |S|-|T|
\end{eqnarray*}
which contradicts to the fact that $S\cup T$ is a connected component of the critical set. Therefore, we have $p_j\le min_j$ for any $j\in V$, which completes the proof of the first part.

It remains to show that $p_j<min_j$ if $j$ is marked for any $j\in V$. If the price of $j$ is marked at the beginning of the considered stage and its price is not increased, by induction assumption, we still have $p_j<min_j$. Thus, it suffices to consider the case when Step~(3.f) of \algmin\ occurs, where all items in $T$ are set to be marked. For this case, we need to show that $p_j<min_j$ for any $j\in T$. Assume otherwise that there is $j_0\in T$ such that $p_{j_0}=min_{j_0}$. It can be seen that that for every $i\in S$, $i$ can get a positive utility from items in $N(i)$, where $N(i)$ is the neighbor set of $i$ in the demand graph at the end of the stage (as if $i$ gets utility 0 from any $j\in N(i)$, then $v_{ij}=p_j$ for every $j\in N(i)$ and Step~(3.e) will be executed).
This is exactly the property used in the above first part to draw a contradiction. Thus, by the same argument, we can show that $p_j<min_j$ for any $j\in T$.

For Step~(5.a) of the algorithm where items in the neighbor set of the critical set of $G'$ are set to be marked as well, its proof of $p_j<min_j$ is the same as above. This completes the proof of the lemma.
\end{proof}


\medskip \noindent \textbf{Lemma~\ref{lemma-min-eq}.} For any given instance of the problem, if \algmin\ outputs $(\mathbf{x},\mathbf{p})$, then it is a competitive equilibrium.

\begin{proof}
This follows by the rule of the algorithm and the construction of the demand graph $G^*$ in Step~(6) --- each $(i,j)\in E$ represents that $i$ gets its maximal utility from $j$ for the given price vector. Note that $\{i\in U~|~N^*(i)\neq \emptyset\} = U^+$, and any $i\in U^+$ is matched to an item in $N^*(i)$ by the algorithm. Note that all items in $N^*(i)$ are either marked or unmarked. If they are marked, the price of each $j\in N^*(i)$ is $p_j+$ by Step~(8.b). Since $u_{ij}(p_j)>0$ and the increment of all marked items from $p_j$ to $p_j+$ is at the same pattern rate, $i$ still prefers item $j$ with a positive utility. Further, if $N^*(i)=\emptyset$, i.e., $i$ is a simpleton vertex, then for any $j\in V$, we have either $v_{ij}\le p_j$ or $b_{ij} \le p_j$, and if $v_{ij} > p_j$ and $b_{ij}=p_j$, $j$ must be marked. That is, $i$ cannot get a positive utility anyway, no matter if $i$ is matched or not. Therefore, everyone is satisfied with the corresponding allocation given price vector $\mathbf{p}$. Moreover, since the matching in Step~(8) has size of $m$, all items are allocated, which implies that the market clearing condition is satisfied. Hence, $(\mathbf{x},\mathbf{p})$ is an equilibrium.
\end{proof}


\medskip \noindent \textbf{Lemma~\ref{lemma-min-exist}.} For any given instance of the problem, if a competitive equilibrium exists, then \algmin\ will output one.

\begin{proof}
Let $(\mathbf{x}^*,\mathbf{p}^*)$ be an equilibrium and $M^*$ be the corresponding matching. By adding dummy buyers as discussed in Section~\ref{section-pre}, all items are allocated. Let $\mathbf{p}=(p_j)_{j\in V}$ be the price vector of the algorithm \algmin\ when it moves to Step~(6) and consider the demand graph $G^*$ defined in Step~(6). By Lemma~\ref{lemma-minj}, we have $p_j\le p^*_j$ for any $j$, and $p_j<p^*_j$ if $j$ is marked.

Let $X=\{i\in U~|~N^*(i)\neq \emptyset\}$, where $N^*(i)$ is the neighborhood of $i$ in $G^*$. By Lemma~\ref{prop-no-critical}, we have $|X'|\le |N^*(X')|$ for any $X'\subseteq X$. In particular, $|X|\le |N^*(X)|$. Consider any $i\in X$. Note that $i$ obtains a positive maximal utility from items in $N^*(i)$ (otherwise, all edges incident to $i$ would be deleted by Step~(3.d)). We claim that $x^*_i\in N^*(i)$ for any $i\in X$, i.e., $i$ wins an item in $N^*(i)$ in the equilibrium $(\mathbf{x}^*,\mathbf{p}^*)$. Assume otherwise that there is $i_0\in X$ such that either $x^*_{i_0}=\emptyset$ or $x^*_{i_0}\notin N^*(i_0)$.  Let
\[Z=\{j\in N^*(X)~|~ j \ \textup{is allocated to a buyer in $U\setminus X$ in the equilibrium $(\mathbf{x}^*,\mathbf{p}^*)$}\}\]
Note that $Z\neq \emptyset$ by the assumption of $i_0$. For each $j\in Z$, let $T(j)\subseteq X$ be set of buyers to which there is an alternating path from $j$ in $G^*$ in terms of matching $M^*$ (i.e., edges are not in $M^*$ and in $M^*$ alternatively), restricted on buyers in $X$ (i.e., all edges in the alternating path are incident to buyers in $X$). By the assumption of $i_0$ and the fact that $|X'|\le |N^*(X')|$ for any $X'\subseteq X$, there must be an item in $Z$, denoted by $j_\ell\in Z$, such that $i_0\in T(j_\ell)$. Assume that the alternating path from $j_\ell$ to $i_0$ is \[(j_\ell,i_{\ell-1},j_{\ell-1},\ldots,i_1,j_1,i_0)\]
That is, $j_1\in N^*(i_0)$, $j_{k+1}\in N^*(i_k)$ and $x^*_{i_k}=j_k$, for $k=1,\ldots,\ell-1$. Suppose that $x^*_{i_\ell}=j_\ell$, i.e., buyer $i_\ell$ wins $j_\ell$ in the equilibrium. Note that $i_\ell\notin X$ by the assumption of $j_\ell\in Z$; thus we must have $v_{i_\ell j_\ell} \le p_{j_\ell}$ or $b_{i_\ell j_\ell}<p_{j_\ell}$. Since $p_{j_\ell}\le p^*_{j_\ell}$, we must have $p^*_{j_\ell}=p_{j_\ell}=v_{i_\ell j_\ell}$ to guarantee that $i_\ell$ obtains a non-negative utility in the equilibrium. Further, note that $i_{\ell-1}$ obtains his maximal positive utility from $j_{\ell}$ in $G^*$, (but $j_\ell$ is allocated to $i_\ell$ in the equilibrium at price $p^*_{j_\ell}=p_{j_\ell}$,) we must have $j_{\ell-1}\in N^*(i_{\ell-1})$ and $p^*_{j_{\ell-1}}=p_{j_{\ell-1}}$. Continue with this argument through the above path and we will conclude that $p^*_{j_1}=p_{j_1}$ eventually. This implies that $i_0$ strictly prefers $j_1\in N^*(i_0)$ to all other items not in $N^*(i_0)$ in the equilibrium, but $i_0$ does not win any item in $N^*(i_0)$, a contradiction.

Hence, for any $i\in X$, $x^*_i\in N^*(i)$. This implies that the matching defined by $\{(i,x^*_i)~|~i\in X\}$ exists in graph $H \supseteq G^*$ (defined in Step~(7) of the algorithm). Let $Y=\{x^*_i~|~i\in X\}$. For any item $j\notin Y$, it must be allocated to a buyer $i\notin X$ in the equilibrium $(\mathbf{x}^*,\mathbf{p}^*)$. If $v_{ij}> p_j$, since $N^*(i)=\emptyset$, we much have $b_{ij}<p^*_j$, which contradicts to the equilibrium condition. Hence, we must have $v_{ij} = p_j = p^*_j$ and $b_{ij}\ge p_j$. Thus, by Step~(7) of the algorithm, edge $(i,j)$ will be added to $H$. Therefore, matching $\{(i,x^*_i)~|~i\in U\}$ exists in graph $H$, which implies that the algorithm will return an allocation and price vector. By Lemma~\ref{lemma-min-eq}, it must be an equilibrium.
\end{proof}


\medskip \noindent \textbf{Lemma~\ref{lemma-polytime}.} The algorithm \algmin\ runs in strongly polynomial time.

\begin{proof}
In each phase, since there are at most $mn$ possible edges in the bipartite graph and each stage will introduce at least one new edge, we will have at most $mn$ stages. Further, we will remove at least one edge between the critical set and its neighbor set for each phase. Since all prices keep increasing, all deleted edges by Steps~(3.e), (3.f) and (5,b) will never be added back. Hence, there are at most $mn$ phases, which implies that there are at most $m^2n^2$ stages in total. In each stage, the algorithm calls a subroutine \increaseprice\ to compute the price level to which one of the conditions (3.b.$\alpha$), (3.b.$\beta$) or (3.b.$\gamma$) is satisfied, and increase prices to that level. Note that by the assumptions of utility functions, subroutine \increaseprice\ can be implemented in strongly polynomial time. Further, by the following characterization shown by Irving~\cite{irving}, which says that given a maximum matching of $G^+=(U^+,V^+;E)$, the critical set of $U^+$ consists of unmatched vertices together with those reachable from them via an alternating path, finding the critical set is equivalent to finding a maximum matching. This gives us the result.
\end{proof}

Finally, we note that the marked items in the final output of the algorithm \algmin\ has an interesting connection to weakly stable matching, as the following claim shows.

\begin{prop}\label{prop-alg-weakly-stable}
In Step~(8.b) of the algorithm, if the price of every marked slot $j$ is not set to be $p_j+$, then the final output $(\mathbf{x},\mathbf{p})$ is not a competitive equilibrium but a weakly stable matching.
\end{prop}
\begin{proof}
Due to Lemma~\ref{lemma-minj}, for any marketed slot $j$, its price $p_j$ is strictly smaller than its price in all equilibria. Hence, the output $(\mathbf{x},\mathbf{p})$ cannot be a competitive equilibrium. To see that $(\mathbf{x},\mathbf{p})$ is a weakly stable matching, consider any buyer $i$ and item $j\neq x_i$. Since matching $\mathbf{x}$ is determined according to the demand graph $H$ defined in Step~(7) of the algorithm, we have either $u_{ix_i}(p_{x_i})\ge u_{ij}(p_{j})$ (if $x_i=\emptyset$, denote $u_{ix_i}(p_{x_i})=0$), or $u_{ix_i}(p_{x_i}) < u_{ij}(p_{j})$ but $b_{ij}=p_j$ (for this case, $j$ must be marked). For the latter case, the seller for item $j$ cannot obtain more revenue from $i$ due to his tight budget constraint. Hence, $(i,j)$ cannot be a weakly blocking pair, which implies that the output $(\mathbf{x},\mathbf{p})$ is a weakly stable matching.
\end{proof}

\section{Minimum Equilibrium Mechanism: Proof of Theorem~\ref{theorem-nash}}\label{appendix-nash}

\begin{proof}
For any fixed truthful bids of utility functions $u_{ij}(\cdot)$ of all other buyers, $i\neq i_0$, we will analyze the strategic behavior of
buyer $i_0$. Assume that the minimum equilibrium is $(\mathbf{x},\mathbf{p})$ when $i_0$ bids $u_{i_0j}(\cdot)$ truthfully
(where the minimum price is either an exact value $p_j$ or notation $p_j+$).
Consider any other possible utility functions $u'_{i_0j}(\cdot)$ that $i_0$ bids. Assume that the minimum equilibrium is
$(\mathbf{x'},\mathbf{p'})$ when $i_0$ changes his bid to $u'_{i_0j}(\cdot)$. It suffices to show that the utility of $i_0$ in
equilibrium $(\mathbf{x},\mathbf{p})$ is larger than or equal to his
utility in $(\mathbf{x'},\mathbf{p'})$. Assume without loss of
generality that $x'_{i_0}=j_0$, i.e., $i_0$ wins item $j_0$ in
equilibrium $(\mathbf{x'},\mathbf{p'})$ (if $i_0$ does not win any
item, then certainly he cannot obtain more utility).

Since $(\mathbf{x},\mathbf{p})$ is an equilibrium, the utility of
every buyer $i$ is maximized at the corresponding allocation $x_i$
given price vector $\mathbf{p}$. That is, if $x_i = j$, then $u_{ij}(p_j)\ge 0$ and for any
other item $j'\neq j$, either $b_{ij'}<p_{j'}$ or $u_{ij}(p_j)\ge
u_{ij'}(p_{j'})$; if $x_i = \emptyset$, then for any item $j$,
either $b_{ij}<p_{j}$ or $v_{ij}\le p_j$. Hence, when $i_0$ bids
$u'_{i_0j}(\cdot)$ untruthfully, to obtain a higher utility,
no matter whether $x_{i_0}=\emptyset$ or $x_{i_0}\neq \emptyset$, it
must be $p'_{j_0}<p_{j_0}$ (if $p_{j_0}$ is with ``+" notation, i.e., $p_{j_0}+$, we have $p'_{j_0}\le p_{j_0}<p_{j_0}+$).
For simplicity, for any buyer $i$, if
$x_i=\emptyset$, we denote $u_{ix_i}(p_{x_i})=0$.

Define
\[T=\{j\in V~|~p'_j < p_j\} \ \ \textup{and} \ \ S=\{i\in U~|~x'_i\in T\}\]
where $S$ is the subset of buyers who win items in $T$ in
equilibrium $(\mathbf{x'},\mathbf{p'})$. That is, $|S|=|T|$ and all
buyers in $S$ win all items in $T$ in $(\mathbf{x'},\mathbf{p'})$.
Note that since $j_0\in T$, $S,T\neq \emptyset$. In particular, this
implies $i_0\in S$. Further, we claim that all buyers in $S$ win all
items in $T$ in equilibrium $(\mathbf{x},\mathbf{p})$ as well.
Otherwise, there is $i\notin S$ and $j\in T$ such that $x_{i}=j$.
Since $u_{ij}(p'_{j}) > u_{ij}(p_{j}) \ge 0$, $i$ must win
an item in equilibrium $(\mathbf{x'},\mathbf{p'})$. Assume that
$x'_i=j'\notin T$. Thus, $b_{ij'}\ge p'_{j'}\ge p_{j'}$. Hence,
\[u_{ij'}(p'_{j'})\stackrel{(1)}{\ge} u_{ij}(p'_{j}) > u_{ij}(p_{j}) \stackrel{(2)}{\ge} u_{ij'}(p_{j'})\]
where (1) follows from equilibrium $(\mathbf{x'},\mathbf{p'})$ and
(2) follows from equilibrium $(\mathbf{x},\mathbf{p})$, a
contradiction to the fact that $p'_{j'}\ge p_{j'}$.

Given $S$ and $T$, define an output $(\mathbf{y},\mathbf{q})$ from
$(\mathbf{x},\mathbf{p})$ and $(\mathbf{x'},\mathbf{p'})$ as
follows:
\begin{itemize}
\item Let $y_i=x'_i$ if $i\in S$, and $y_i=x_i$ if $i\notin S$.

\item Let $q_j=p'_j$ if $j\in T$, and $q_j=p_j$ if $j\notin T$.
\end{itemize}
Since all buyers in $S$ win in all items in $T$ in both equilibria
$(\mathbf{x},\mathbf{p})$ and $(\mathbf{x'},\mathbf{p'})$, the
allocation defined by $\mathbf{y}$ is feasible. Note that for any
item $j$, $q_j\le p_j$ (and $q_j<p_j$ if $j\in T$). We next analyze the utility of every
buyer in $(\mathbf{y},\mathbf{q})$.
\begin{itemize}
\item For any buyer $i\notin S$, we have $x_i,x'_i\notin T$. By equilibrium $(\mathbf{x},\mathbf{p})$, $i$ weakly prefers $y_i$ to all other items in $V\setminus T$ in $(\mathbf{y},\mathbf{q})$. For any item $j\in T$, we have either $b_{ij}<p'_j=q_j$ or
    \[u_{iy_i}(q_{y_i}) = u_{ix_i}(p_{x_i}) \stackrel{(1)}{\ge} u_{ix'_i}(p_{x'_i}) \stackrel{(2)}{\ge} u_{ix'_i}(p'_{x'_i}) \stackrel{(3)}{\ge} u_{ij}(p'_{j}) = u_{ij}(q_{j}) \ \ \ (*)\]
    where (1) follows from equilibrium $(\mathbf{x},\mathbf{p})$, (2) follows from $p'_{x'_i}\ge p_{x'_i}$ as $x'_i\notin T$, and (3) follows from equilibrium $(\mathbf{x'},\mathbf{p'})$. Note that in $(*)$, if $x_i=\emptyset$ or $x'_i=\emptyset$, we can build the same series of inequalities and show that $u_{ij}(q_j)\le 0$. That is, $i$ (weakly) prefers his allocation $y_i$ to all items in $T$ in output $(\mathbf{y},\mathbf{q})$.\footnote{If there were no budget constraints, we could easily get a contradiction to the minimum equilibrium $(\mathbf{x},\mathbf{p})$ at this point --- construct a new equilibrium $(\mathbf{x},\mathbf{p}^*)$ from $(\mathbf{x},\mathbf{p})$ by reducing all prices $p_j$ for $j\in T$ a very small amount in a similar (converse) manner as Proposition~\ref{prop-step-8b} to keep the same demand preference. With budget constraints, however, $(\mathbf{x},\mathbf{p}^*)$ may not be an equilibrium: buyers in $S$ may strictly prefer other items in $T$ to $x_i$ since they become to have enough budgets to some more preferred items.}

\item For any buyer $i\in S$, $i\neq i_0$, we have $x_i,x'_i\in T$. By equilibrium $(\mathbf{x'},\mathbf{p'})$, $i$ weakly prefers $y_i$ to all other items in $T$ in $(\mathbf{y},\mathbf{q})$. For any item $j\notin T$, either $b_{ij}<p_j=q_j$ or
    \[u_{iy_i}(q_{y_i}) = u_{ix'_i}(p'_{x'_i}) \stackrel{(1)}{\ge} u_{ix_i}(p'_{x_i}) > u_{ix_i}(p_{x_i}) \stackrel{(2)}{\ge} u_{ij}(p_j) = u_{ij}(q_j) \ \ \ (\#)\]
    where $(1)$ follows from equilibrium $(\mathbf{x'},\mathbf{p'})$, and $(2)$ follows from equilibrium $(\mathbf{x},\mathbf{p})$. That is, $i$ strictly prefers his allocation $y_i$ to all items in $V\setminus T$ in output $(\mathbf{y},\mathbf{q})$.
\end{itemize}

It remains to consider the utility of $i_0$ in output $(\mathbf{y},\mathbf{q})$. Note that the equilibrium $(\mathbf{x'},\mathbf{p'})$ is computed in terms of the false bid $u'_{i_0j}(\cdot)$ from buyer $i_0$, whereas the true utility of $i_0$
should be computed in terms of $u_{i_0j}(\cdot)$ rather than $u'_{i_0j}(\cdot)$. Thus it
is possible that $i_0$ strictly prefers other items to $y_{j_0}$ in
$(\mathbf{y},\mathbf{q})$ in terms of his true utility functions $u_{i_0j}(\cdot)$. In the following we adjust allocations
and prices for items in $T$ in $(\mathbf{y},\mathbf{q})$ so as to
satisfy $i_0$ as well.

Let $G$ be the real demand bipartite graph (i.e., the utility of $i_0$ is in terms of $u_{i_0j}$) with respect to price vector
$\mathbf{q}$ restricted on the sub-instance given by $S$ and $T$
(i.e., ignore all buyers $U\setminus S$ and items $V\setminus T$).
Let $S'\subseteq S$ be the critical set of buyers of $G$ and
$T'=N(S')$, where $N(S')$ is the neighbor set of $S'$ of $G$. We recursively
increase all prices $q_j$ for $j\in T'$ continuously and
simultaneously according to subroutine \increaseprice. In the process of increasing prices,
the demand bipartite graph $G$, as well as the critical set $S'$ and
its neighbor set $T'$, will be updated dynamically with respect to
the current price vector $\mathbf{q}$. Note that when prices
increase, items $T\setminus T'$ might enter into the demand
set of $i\in S'$, and an edge $(i,j)$ where $i\in S'$ and $j\in T'$
might be broken due to budget constraint. The process stops when the
critical set of $G$ becomes $\emptyset$.

We claim that at the end of the process when $G$ has no
non-empty critical set, the price of every item $j\in T$ has $q_j < p_j$.
Otherwise, consider the first implementation of subroutine \increaseprice\
where there is $j\in T$ whose price $q_{j}$ is increased to be at least $p_{j}$
at the end of the implementation. By Proposition~\ref{prop-price-increase}, the implementation of \increaseprice\ can
be viewed as a process of increasing prices continuously. Consider the first moment when there is $j\in T$ whose
price $q_{j}$ is increased exactly to $p_{j}$ in the process; and consider the price vector $\mathbf{q}=(q_j)_{j\in T}$, demand graph $G$, critical set $S'$ and its neighbor
set $T'$ at that moment. Let $V'=\{j\in T~|~q_j=p_j\}$ and
$U'=N(V')$ be the neighbor set of $V'$. Since only prices of items
in $T'$ are increased, we have $V'\subseteq T'$. By the definition
of the critical set, we have $|U'|>|V'|$. Otherwise,
\[|S'|-|T'|=|S'\setminus U'| + |S'\cap U'| - |T'\setminus V'| - |V'| \le |S'\setminus U'| - |T'\setminus V'| \le |S'\setminus U'| - |N(S'\setminus U')|\]
where the last inequality follows from $N(S'\setminus U')\subseteq
T'\setminus V'$, a contradiction. Hence, there is $i\in U'$ such
that $i$ wins an item in $T\setminus V'$ in equilibrium
$(\mathbf{x},\mathbf{p})$, i.e., $x_i\in T\setminus V'$. Therefore,
for any $j\in N(i)\cap V'$, we have
\[u_{ij}(p_j)=u_{ij}(q_j) \stackrel{(1)}{\ge} u_{ix_i}(q_{x_i}) \stackrel{(2)}{>} u_{ix_i}(p_{x_i})\]
where $(1)$ follows from the demand set of $i$, and $(2)$ follows from the fact that $q_{x_i}<p_{x_i}$ when $x_i\in T\setminus V'$.
This contradicts to the fact that $(\mathbf{x},\mathbf{p})$ is an equilibrium.

At the end of the increment process when $G$ has no non-empty critical set, we know that for
any subset $U\subseteq S$, $|U|\le |N(U)|$. By Hall's theorem and the fact that $|S|=|T|$, there is a
perfect matching $M$ between $S$ and $T$. We define an output $(\mathbf{x}^*,\mathbf{p}^*)$ from
$(\mathbf{y},\mathbf{q})$, where $\mathbf{q}$ is the price vector at the end of the above increment process:
allocations of buyers in $S$ are defined
according to $M$ and allocations of buyers in $U\setminus S$ remain the same as
$\mathbf{y}$; and $\mathbf{p}^*=\mathbf{q}$. Similar to the
inequalities established in $(*)$ and $(\#)$ above, we can show that
all buyers get their utility-maximal allocation in
$(\mathbf{x}^*,\mathbf{p}^*)$. Hence, $(\mathbf{x}^*,\mathbf{p}^*)$
is an equilibrium. This contradicts to the fact that
$(\mathbf{x},\mathbf{p})$ is a minimum equilibrium.

Therefore, we have $p'_{j_0}\ge p_{j_0}$, i.e., buyer $i_0$ will
never obtain more utility from bidding untruthfully. This implies
that it is of best interest to bid truthfully for the minimum
equilibrium output.
\end{proof}

\end{document}